\documentclass[a4paper,onecolumn,11pt,accepted=2025-07-03]{quantumarticle}

\def\papershowtoc{}
\pdfoutput=1
\usepackage[final]{hyperref}
\hypersetup{
           breaklinks=true,   
           colorlinks=true,   
           citecolor=red,
           allcolors={red}
        }
\usepackage[utf8]{inputenc}
\usepackage{amsmath}  
\usepackage{amsfonts}
\usepackage{amssymb}
\usepackage{amsthm}
\usepackage{braket}
\usepackage{algorithm2e}
\usepackage[]{nicefrac}
\usepackage{bbm}
\usepackage{qcircuit}
\usepackage{color}
\usepackage[normalem]{ulem}
\usepackage{mathtools}
\newcommand{\defeq}{\coloneqq}

\usepackage{cancel}
\usepackage{graphicx}
\usepackage{varwidth}
\usepackage{thm-restate}
\usepackage{caption}

\usepackage{subcaption}
\usepackage{tikz}
\usetikzlibrary{backgrounds,fit,decorations.pathreplacing,calc}

\usepackage{complexity}
\newclass{\sharpP}{\#P}
\newclass{\MAXproblem}{MAX}

\newcommand{\ketbra}[2]{|#1\rangle\! \langle #2|}

\newcommand{\norm}[1]{\left\lVert #1 \right\rVert}

\newcommand{\pvp}{\vec{p}{\kern 0.45mm}'}
\let\oldnabla\nabla
\renewcommand{\nabla}{\oldnabla\!}

\def\Tr{\mathrm{Tr}}

\DeclareMathOperator{\diff}{d \!}

\providecommand{\od}[3][]{\ensuremath{
\ifinner
\tfrac{\diff{^{#1}}#2}{\diff{{#3}^{#1}}}
\else
\dfrac{\diff{^{#1}}#2}{\diff{{#3}^{#1}}}
\fi
}}

\long\def\ignore#1{}

\newtheorem{theorem}{Theorem}

\newtheorem{lemma}[theorem]{Lemma}

\newtheorem{definition}[theorem]{Definition}

\usepackage{qcircuit}

\begin{document}

\title{Unstructured Adiabatic Quantum Optimization: Optimality with Limitations}

\author{Arthur Braida}
\email{arthur.braida@ulb.be}
\affiliation{QuIC, Ecole Polytechnique de Bruxelles, Universit\'{e} libre de Bruxelles, Brussels, Belgium}

\author{Shantanav Chakraborty}
\email{shchakra@iiit.ac.in}
\affiliation{CQST and CSTAR, International Institute of Information Technology Hyderabad, Telangana, India} 

\author{Alapan Chaudhuri}
\email{alapan.chaudhuri@research.iiit.ac.in}
\affiliation{CQST and CSTAR, International Institute of Information Technology Hyderabad, Telangana, India} 

\author{Joseph Cunningham}
\email{joseph.cunningham@ulb.be}
\affiliation{QuIC, Ecole Polytechnique de Bruxelles, Universit\'{e} libre de Bruxelles, Brussels, Belgium}

\author{Rutvij Menavlikar}
\email{rutvij.menavlikar@research.iiit.ac.in}
\affiliation{CQST and CSTAR, International Institute of Information Technology Hyderabad, Telangana, India} 

\author{Leonardo Novo}
\email{leonardo.novo@inl.int}
\affiliation{International Iberian Nanotechnology Laboratory, Braga, Portugal}

\author{J\'{e}r\'{e}mie Roland}
\email{jeremie.roland@ulb.be}
\affiliation{QuIC, Ecole Polytechnique de Bruxelles, Universit\'{e} libre de Bruxelles, Brussels, Belgium}

\begin{abstract}
In the circuit model of quantum computing, amplitude amplification techniques can be used to find solutions to NP-hard problems defined on $n$-bits in time $\text{poly}(n) 2^{n/2}$. In this work, we investigate whether such general statements can be made for adiabatic quantum optimization, as provable results regarding its performance are mostly unknown. Although a lower bound of $\Omega(2^{n/2})$ has existed in such a setting for over a decade, a purely adiabatic algorithm with this running time has been absent. We show that adiabatic quantum optimization using an unstructured search approach results in a running time that matches this lower bound (up to a polylogarithmic factor) for a broad class of classical local spin Hamiltonians. For this, it is necessary to bound the spectral gap throughout the adiabatic evolution and compute beforehand the position of the avoided crossing with sufficient precision so as to adapt the adiabatic schedule accordingly. However, we show that the position of the avoided crossing is approximately given by a quantity that depends on the degeneracies and inverse gaps of the problem Hamiltonian and is NP-hard to compute even within a low additive precision. Furthermore, computing it exactly (or nearly exactly) is \#P-hard. Our work indicates a possible limitation of adiabatic quantum optimization algorithms, leaving open the question of whether provable Grover-like speed-ups can be obtained for any optimization problem using this approach. 
\end{abstract}	
\maketitle
\newpage
\ifdefined\papershowtoc
\newpage
\tableofcontents
\newpage
\fi
\section{Introduction}
\label{sec:intro}
Adiabatic Quantum Computation (AQC) is an interesting, Hamiltonian-based alternative to the standard gate-based model of quantum computation \cite{farhi2000adiabatic, farhi2001adiabatic}. Being physically motivated, this framework has deep connections to condensed matter physics \cite{hastings2013obstructions, albash2018adiabatic}, to complexity theory \cite{kempe2006complexity}, and is also the zero temperature limit of quantum annealing \cite{johnson2011quantum}. In fact, AQC is a universal quantum computational model: the circuit and the adiabatic models are equivalent up to a polynomial overhead \cite{aharonov2007adiabatic}. Consequently, over the years, it has been instrumental to the design of several novel quantum algorithms \cite{aharonov2003stategeneration, krovi2010adiabatic, somma2012quantum, garnerone2012pagerank, Hastings2021powerofadiabatic, gilyen2021subexponential}. Often, the adiabatic model also serves as a tool for developing quantum algorithms in the circuit model \cite{subacsi2019qlsadiabatic, lin2020optimalpolynomial, an2022qlstimeoptimal, dalzell2023mind}: a discretized approximation of the adiabatic evolution can be simulated in the gate model using standard techniques such as Hamiltonian simulation \cite{berry2020timedependent} and phase randomization \cite{boixo2009eigenpath}.

The underlying principle behind AQC is the adiabatic theorem of quantum mechanics \cite{jansen2007bounds}, an idea that is quite distinct from the circuit model. The computation starts from the known ground state of an initial Hamiltonian $H_0$, which is easy to prepare (usually a product state). This Hamiltonian is then transformed ``slowly'' (adiabatically) into a final Hamiltonian $H_P$, whose ground states encapsulate the solution to the underlying computational problem. The total Hamiltonian is an interpolation between the initial and the final Hamiltonians, i.e.\ $H(s)=(1-s)H_0+sH_P$, where $s:[0, T]\mapsto [0,1]$, known as the adiabatic ``\textit{schedule}", determines the adiabatic path from $H_0$ to $H_P$, while $T$ is the total time of evolution. The quantum adiabatic theorem guarantees that the final state has a large overlap with the desired ground state, provided $T$ (which is also the algorithmic running time) is at least a polynomial in the inverse of the minimum spectral gap of any intermediate Hamiltonian $H(s)$ along the adiabatic path \cite{jansen2007bounds, elagrt2012note}. 

Originally, AQC was formulated as a generic method for efficiently solving classically hard optimization problems, known as adiabatic quantum optimization (AQO) \cite{farhi2000adiabatic, farhi2001adiabatic, reichardt2004adiabatic}. Indeed, AQO provides a natural framework to solve \NP-hard problems by finding the minimum of a cost function encoded in the ground states of an $n$-qubit Ising Hamiltonian \cite{barahona1982computational, lucas2014ising}. Provable results about the performance of AQO in such settings are largely unknown, as it becomes difficult to compute the spectral gap throughout the adiabatic evolution. Exponential speedups are unlikely as for random instances of certain \NP-hard problems, exponentially small gaps appear, and the system gets stuck in one of the many local minima, leading to a running time that can be slower than even classical brute force search \cite{altshuler2010anderson}. In this regard, a natural question to ask is whether it is possible to prove at least a Grover-like speedup \cite{roland2004quantum} over unstructured classical search approaches for the problem of finding the global minimum of a cost function. Note that this is possible in the circuit model: Grover's algorithm (or more precisely, quantum amplitude amplification) can be leveraged to find the minimum of a cost function encoded in any Ising Hamiltonian in time $\Theta(2^{n/2}\poly(n))$ \cite{ambainis2004quantumsearch}. It is thus plausible to expect that such a generic result would also be possible in the adiabatic setting, given that it is a universal model of quantum computation. Moreover, in order to solve this problem, it is reasonable to assume that all that should be required is access to an adiabatic quantum optimizer and a classical computer, such that both are allowed to run for a time of at most $O(2^{n/2}~\poly(n))$. How can then such a general result be achieved in the context of AQO? Note that naïvely encoding a circuit model algorithm as an AQC already requires a polynomial overhead, erasing any speed-up that is quadratic. To tackle this problem, it is thus natural to directly employ the adiabatic version of Grover's algorithm \cite{roland2004quantum}. In the adiabatic setting, it is known that the use of unstructured quantum search leads to a lower bound of $\Omega(2^{n/2})$ for any adiabatic algorithm \cite{farhi2008fail}. However, to the best of our knowledge, it is unknown whether any AQO algorithm can attain the aforementioned lower bound for finding ground states of Hamiltonians encoding \NP-hard problems. Indeed, as mentioned previously, this has been an outstanding problem primarily because it is difficult to bound the spectral gap throughout the adiabatic evolution, which is crucial for determining the running time of any adiabatic algorithm.

In this work, we provide an adiabatic algorithm based on unstructured quantum search that can find the minimum of an Ising Hamiltonian in time $O(2^{n/2}~\poly(n))$, matching the aforementioned lower bound (up to a factor of poly$(n)$). The results are quite general: our algorithm can find the minimum of any Hamiltonian that is diagonal in the computational basis, provided it has a sufficiently large spectral gap. As with the adiabatic version of Grover's algorithm, there exists only a single avoided crossing between the two lowest eigenstates of the adiabatic Hamiltonian. However, the overall spectrum is significantly more complicated, and the position of the avoided crossing is non-trivial. Moreover, prior knowledge of the position of this avoided crossing is crucial to constructing a local schedule that is mindful of the instantaneous spectral gap of the time-dependent Hamiltonian $H(s)$: the adiabatic algorithm uses this information in the adaptation of the adiabatic schedule, which can be fast in regions of higher gap and slow in regions of smaller gap \cite{vandam2001powerful, roland2004quantum}. We (i) rigorously bound the spectral gap of $H(s)$ for any $s\in [0,1]$, and (ii) derive a closed-form expression that approximates the position of the avoided crossing with sufficient accuracy. In fact, we demonstrate that it suffices for the approximation to be (roughly) within an additive error of $2^{-n/2}$ of the actual position of the avoided crossing. Both (i) and (ii) allow us to construct the appropriate local adiabatic schedule that leads to the optimal running time. 

As mentioned previously, in order to construct the optimal schedule, prior knowledge of the position of the avoided crossing is crucial. Unlike the adiabatic version of Grover's algorithm, the position of the avoided crossing is a function of the degeneracies and eigenvalue gaps in the spectrum of the final Hamiltonian. We prove that it is hard to approximate this quantity to even a precision that is much larger: any classical procedure that predicts this quantity up to an additive precision of $1/\poly(n)$ can be used as an oracle to solve any problem in the complexity class \NP, i.e.\ it is \NP-hard. Indeed, we prove that only a constant number of queries to such a classical algorithm is enough to solve the Boolean satisfiability problem (in particular 3-\SAT), a well-known \NP-complete problem \cite{Karp1972}. On the other hand, estimating this quantity exactly (or nearly exactly, i.e.\ up to an additive precision of $2^{-\poly(n)}$) is \sharpP-hard. Recall that \sharpP\ comprises of problems for which it is possible to count the number of solutions in polynomial time, the counting analogue of \NP\ \cite{valiant1979complexity}. \sharpP-hard (or \sharpP-complete) problems are also notoriously difficult to solve, and this has been leveraged to prove the hardness of sampling procedures at the heart of a number proposals for demonstrating quantum advantage, namely Boson sampling \cite{aaronson2011bosonsampling}, IQP sampling \cite{bremner2017achievingquantum}, and random circuit sampling \cite{bouland2019complexity, movassagh2023hardness}. These complexity-theoretic reductions prove that in order to optimally implement AQO, it is necessary to solve a computationally hard problem beforehand. It is thus unlikely that a generic classical algorithm will approximate the position of the avoided crossing in $\widetilde{O}(2^{n/2})$ time.

This points to a fundamental limitation of the adiabatic framework, which is absent in the circuit model. In the latter, it is possible to simply prepare the ground state of the Ising Hamiltonian by starting from an initial equal superposition state and applying standard techniques such as quantum phase estimation and amplitude amplification \cite{abrams1999quantum} (modern techniques such as LCU \cite{ge2019faster, chakraborty2024implementing} or QSVT \cite{lin2020nearoptimalground} to a block encoding of the underlying Hamiltonian \cite{chakraborty2019power} also works). Thus, our work leaves open the question of whether this limitation may be overcome when one only has access to a device operating in the adiabatic setting (along with a classical computer). More generally, can we develop a purely adiabatic algorithm providing Grover-like speedups for the minimum finding problem without (i) needing access to a digital quantum computer and (ii) having to solve a computationally hard problem in the process? This could be possible by a suitable modification of the adiabatic Hamiltonian (such as by adding extra qubits) or by introducing intermediate Hamiltonians along the adiabatic path so that the position of the avoided crossing does not depend on the spectrum of the problem Hamiltonian. 

This article is organized as follows. In the rest of this section, we review some basic definitions used in this article in Sec.~\ref{subsec:prelim}, provide a brief overview of our results in Sec.~\ref{subsec:results}, and discuss related works in Sec.~\ref{subsec:prior-work}. In Sec.~\ref{sec:poac-and-schedule}, we develop the adiabatic quantum algorithm for finding the minimum of a cost function encoded in an Ising Hamiltonian. We keep our analysis general, as it holds for any Hamiltonian diagonal in the computational basis (with a sufficiently large spectral gap). To this end, we bound the spectral gap of the adiabatic Hamiltonian and find a closed-form expression approximating the position of the avoided crossing in Sec.~\ref{subsec:poac}. In Sec.~\ref{subsec:schedule-run-time}, we derive the algorithmic running time. The computational hardness of predicting the position of the avoided crossing is discussed in Sec.~\ref{sec:hardness-avoided-crossing}. Sec.~\ref{sec:np-hardness-A1} deals with proving that approximating this quantity to even a low precision is \NP-hard while in Sec.~\ref{sec:sharp-P-hardness-A1}, we show that (nearly) exactly estimating it is \sharpP-hard. Finally, we summarize our results and discuss possible open problems in Sec.~\ref{sec:discussion}

\subsection{Preliminaries}
\label{subsec:prelim}
In this section, we will briefly introduce some preliminary concepts that we will use in the rest of this article. We begin by stating the notation we shall be using throughout the article.
~\\~\\
\noindent \textbf{Complexity theoretic notations:~}Throughout the article, we shall be using the standard complexity-theoretic notations. The \textit{Big-O} notation, $g(n)= O(f(n))$ or $g(n)\in O(f(n))$, implies that $g$ is upper bounded by $f$. That is, there exists a constant $k>0$ and a real number $n_0$ such that $g(n)\leq k\cdot f(n)$ for all $n\geq n_0$. The \textit{Big-Omega} notation, $g(n)=\Omega(f(n))$ (or equivalently $g(n)\in \Omega(f(n))$), implies $g(n)\geq k f(n)$ ($g$ is lower bounded by $f$). The \textit{Theta} notation is used when $g$ is both bounded from above and below by $f$, i.e.\ $g(n)=\Theta(f(n))$ (or $g(n)\in\Theta(f(n))$), if $g(n) = O(f(n))$ and $g(n) = \Omega(f(n))$.

For each of these notations, it is standard to use \textit{tilde} ($\sim$) to hide polylogarithmic factors. For instance, $\widetilde{O}(f(n))=O(f(n)\polylog(f(n)))$. This applies to the other notations as well. Often, we shall use the notation $\poly(f(n))$ to denote a function that is some polynomial of $f(n)$.
~\\~\\
\textbf{Norm:~}Unless otherwise specified $\norm{A}$ will denote the spectral norm of operator $A$ while $\norm{\ket{v}-\ket{u}}$ will denote the $\ell_2$-norm distance between quantum states $\ket{v}$ and $\ket{u}$.
~\\~\\
\textbf{Resolvent of an operator:~}For all $\gamma \in \mathbb{C}$ and normal operators $A$ (which include the self-adjoint operators), the distance between $\gamma$ and the closest point in the spectrum of $A$ is given by $\norm{R_A(\gamma)}^{-1}$, where 
\begin{equation}
\label{eq:resolvent-definition}
R_A(\gamma) \defeq (\gamma \cdot I - A)^{-1},
\end{equation} 
is the resolvent and $I$ is the identity operator.
~\\~\\
\textbf{Sherman-Morrison formula:~}Suppose $A\in\mathbb{C}^{N\times N}$ is an invertible square matrix and $\ket{u},\ket{v}\in\mathbb{C}^N$ are column vectors such that $1+\bra{v}A^{-1}\ket{u}\neq 0$. Then, the Sherman-Morrison formula \cite{sherman_morrison} states the following:
\begin{equation}
\label{eq:sherman-morrison}
\left(~A+\ketbra{u}{v}~\right)^{-1}= A^{-1}-\dfrac{A^{-1}\ketbra{u}{v}A^{-1}}{1+\bra{v}A^{-1}\ket{u}}
\end{equation}
~\\
Next, we discuss some complexity classes that are crucial for our results, while we refer the readers to Ref.~\cite{arora2009computational} for details. We begin by defining the class \NP: 
~\\~\\
\textbf{The complexity class \NP:~} The complexity class \NP\ is the set of decision problems (with a $0/1$ output) for which the output can be verified efficiently (in polynomial time). An equivalent definition is that a problem is in \NP\  if a non-deterministic Turing machine decides the problem in polynomial time. It can be formally defined as follows:

\begin{definition}
The complexity class \NP\ is the set of all problems that are decided by a non-deterministic polynomial-time Turing machine.
\end{definition}  
Interestingly, there is a subset of problems in \NP\ that are at least as hard as any other problem in the class. These are the so-called \NP-hard problems. Indeed, any problem in \NP\ can be efficiently reduced to a \NP-hard problem. Moreover, if it can be established that a given \NP-hard problem also belongs to the class \NP, then it is \NP-complete. Classical computers are not expected to solve \NP-hard (or \NP-complete) problems efficiently. Next, we define the Boolean satisfiability (\SAT) problem, which is arguably the most famous \NP-complete problem: 
~\\~\\
\textbf{The 3-\SAT\ problem:~}  The Boolean satisfiability (\SAT) problem can be stated as follows: given a Boolean formula $F$, does there exist an assignment of truth values $(0/1)$ to the variables of $F$, so that $F$ evaluates to true? This is a decision problem, i.e.\ given a formula as an input, the algorithm is supposed to output a binary (0/1) answer. In the most common version of this problem, the formula $F$ is comprised of multiple clauses, with each clause containing some fixed number of Boolean variables and their negations (known as literals), separated by Boolean ORs (denoted by $\lor$). The clauses are joined together by Boolean AND (denoted by $\land$). Then k-\SAT\ refers to satisfiability problems where there are exactly $k$ literals per clause (for 3-\SAT, we have $k=3$). The satisfiability problem, and in particular 3-\SAT, was the first problem in \NP, proven to be \NP-complete \cite{Karp1972}. 
~\\~\\
\textbf{The complexity class \textbf{\sharpP}:~}The class \sharpP, introduced by Valiant \cite{valiant1979permanent, valiant1979complexity}, can be seen as the counting analogue of~\NP. For problems in \NP, for any given input, we ask whether there is a solution, i.e.\ the output is binary $(0/1)$ (decision problem). On the other hand, in the case of counting problems, for any given input, the output is a natural number denoting the number of solutions. We formally state the definition of this class below: 

\begin{definition}
The complexity class \sharpP\ is the set of all functions $f:\{0,1\}^*\mapsto\mathbb{N}$ such that there is a non-deterministic polynomial-time Turing Machine $M$ such that for all $x\in \{0,1\}^*$, $f(x)$ denotes the number of accepting branches of $M$.
\end{definition}

Much like \NP\, it is possible to analogously define the complexity classes \sharpP-hard\ and \sharpP-complete. Problems in this class are also hard for a classical computer to solve efficiently. For instance, \#3-SAT (the counting version of 3-\SAT) asks for the number of satisfying assignments in a given Boolean formula and is known to be \sharpP-complete.
~\\~\\
\textbf{Adiabatic Quantum Computation:~}AQC is an interpolation between two $n$-qubit Hamiltonians, \(H_0\) and \(H_P\), such that the evolution is governed by the time-dependent Hamiltonian 
$$
H(s)=(1-s)H_0+s H_P,
$$
where \(s(t): [0, T] \rightarrow [0, 1]\) is a monotonic function referred to as the adiabatic ``schedule" \cite{albash2018adiabatic}. The system begins in the ground state of \(H_0\), which is unique and typically easy to prepare (such as a product state) and evolves (adiabatically) over a long enough time \(T\) such that, by the end of the evolution, the final state is \(\varepsilon\)-close to the ground state of \(H_P\). The precise meaning of ``a long enough time" $T$ is captured by the rigorous version of the quantum adiabatic theorem \cite{jansen2007bounds, elagrt2012note}.  We state the following result from Ref.~\cite{jansen2007bounds} in the finite-dimensional case:

\begin{lemma}[Adiabatic Theorem, Theorem 3 of Ref.~\cite{jansen2007bounds}]
\label{lem:adiabatic-theorem}
Suppose $H(s)$ is a Hamiltonian with a bounded norm that is twice differentiable (denoted by $H'(s)$ and $H''(s)$ respectively). Consider a closed system evolving under $H(s)$ given by
$$
\dfrac{i}{T}\frac{\partial}{\partial s}\ket{\psi(s)} = H(s)\ket{\psi(s)},
$$
such that $\ket{\psi(0)}$ is a ground state of $H(0)=H_0$ and $0\leq s\leq 1$, and $T$ is the total evolution time. Furthermore, let $P(s)$ be the projection on to the space spanned by the ground states of $H(s)$, which has degeneracy $d$, and $g(s)$ be the spectral gap of $H(s)$. Then, there exists a constant $C>0$ such that the quantum adiabatic theorem satisfies
$$
\left|1-\braket{\psi(s)|P(s)|\psi(s)}\right|\leq \nu^2(s),
$$
where
$$
\nu(s)= C\left\{\dfrac{1}{T}\dfrac{d\norm{H'(0)}}{g(0)^2}+\dfrac{1}{T}\dfrac{d\norm{H'(s)}}{g(s)^2}+\dfrac{1}{T}\int_{0}^{s} \left( \dfrac{d\norm{H''(s)}}{g(s)^2}+ \dfrac{d^{3/2}\norm{H'(s)}}{g(s)^3}\right)~ds\right\}.
$$
\end{lemma}

This lemma provides a precise lower bound on the total evolution time $T$ of the underlying AQC. If, at the end of the adiabatic evolution, we wish to end up in a state that is $\varepsilon$-close in trace distance to the underlying ground states, it is enough to make sure that $\nu^2(1)=\varepsilon$, which can be achieved by choosing some $T$ which is a polynomial in the inverse of the minimum spectral gap of any intermediate Hamiltonian along the adiabatic path, and a polynomial in $1/\varepsilon$. That is,  
$$
T\geq \poly\left(\dfrac{1}{\min_s g(s)}~,~\dfrac{1}{\varepsilon} \right).
$$
In Sec.~\ref{sec-app:adiabatic-theorem}, we develop a version of the adiabatic theorem, which requires fewer assumptions on the adiabatic Hamiltonian and is simpler to analyze. We leverage these results to derive the running time of AQO. We also make the polynomial dependence on the gap and the error more precise.  
~\\~\\
\textbf{Adiabatic Quantum Optimization:~}AQO provides a natural framework to solve hard classical optimization problems, which was the original motivation behind the adiabatic model of quantum computation \cite{farhi2001adiabatic}. The task involves finding the minimum of a cost function encoded in an $n$-qubit local classical spin Hamiltonian (such as the Ising spin glass Hamiltonian). As mentioned previously, adiabatic quantum optimization involves evolving the time-dependent Hamiltonian $H(s)=(1-s)H_0+sH_P$ from the ground state of $H_0$ to the ground state of the problem Hamiltonian $H_P$, for a total time $T$ according to the schedule $s:[0, T]\mapsto [0,1]$. Since we are interested in solving optimization problems via unstructured adiabatic search, the initial Hamiltonian 
$$
H_0=-\ket{\psi_0}\bra{\psi_0},
$$  
where $\ket{\psi_0}=\ket{+}^{\otimes n}$, is the equal superposition over all computational basis states \cite{roland2004quantum, farhi2008fail}. For the problem Hamiltonian, we consider any $n$-qubit Hamiltonian diagonal $H_z$ that is diagonal in the computational basis. Without loss of generality, we assume that $H_z$ is appropriately rescaled and normalized such that its eigenvalues lie in $[0,1]$. Suppose the spectral decomposition of
$$
H_{z}=\sum_{z\in \{0,1\}{^n}} E_z\ket{z}\bra{z},
$$ 
where $E_z$ is the eigenvalue and $\ket{z}$ is the corresponding eigenvector. Suppose that $H_z$ has $M$ distinct eigenlevels with eigenvalues $0\leq E_0<E_1<\ldots<E_{M-1}\leq 1$, such that energy level with eigenvalue $E_k$ is $d_k$-degenerate. Formally, for $0\leq k\leq M-1$, we define a set of bit-strings
$$
\Omega_k=\left\{z~|~z\in\{0,1\}^n,\ H_z\ket{z}=E_k\ket{z}\right\},
$$
such that $|\Omega_k|=d_k$ is the degeneracy of eigenvalue $E_k$ and $\sum_{k} d_k = 2^n=N$. Furthermore, we require that our problem Hamiltonians satisfy a certain spectral condition. For this, we define the following spectral functions of the problem Hamiltonian: 
\begin{equation}
\label{eq:spectral-parameters}
A_p=\dfrac{1}{N}\sum_{k=1}^{M-1} \dfrac{d_k}{\left(E_k-E_0\right)^p}, ~~~\text{ where }p\in\mathbb{N}.
\end{equation}
These parameters, which are functions of inverse eigenvalue gaps and degeneracies of the spectrum of the problem Hamiltonian, will also become important for (i) predicting the position of the avoided crossing, (ii) determining the algorithmic running time, and (iii) proving hardness results.

\subsection{Summary of our results}
\label{subsec:results}

In this section, we state the main results of this work and outline the techniques central to their derivation. We begin by formally defining the class of problem Hamiltonians $H_z$ for which our results hold.
\begin{definition}[The problem Hamiltonian] 
\label{def:prob-Ham} 
Let $N=2^n$, and $c \ll 1$ is a small enough positive constant. We consider any $N\times N$ Hamiltonian $H_z$ that is diagonal in the computational basis with $M~(\leq N)$ distinct eigenvalues of the Hamiltonian denoted as $0\leq E_0< E_1 < \cdots E_{M-1}\leq 1$ such that the eigenlevel corresponding to eigenvalue $E_k$ is $d_k$-degenerate, where each $d_k$ is a non-negative integer satisfying $\sum_{k=0}^{M-1} d_k = N$. Additionally, if $\Delta=E_1-E_0$ is the spectral gap and $A_2$ is as defined in Eq.~\eqref{eq:spectral-parameters}, we require that the following condition holds for the  spectrum of $H_z$:
\begin{equation}
\label{eq:spectral-condition}
\dfrac{1}{\Delta}\sqrt{\dfrac{d_0}{A_2 N}}< c.
\end{equation} 
\end{definition}

The overall adiabatic Hamiltonian is then defined as
\begin{equation}
\label{eq:general-H(s)}
H(s)=-(1-s)\ket{\psi_0}\bra{\psi_0}+s H_z.
\end{equation}
Note that $H_z$ encompasses a large family of classical local spin Hamiltonians. The spectral condition is satisfied by any $H_z$ with a large enough spectral gap. Indeed, observe that $A_2\geq 1-1/N$, implying that our results hold whenever $\Delta > \frac{1}{c}\sqrt{d_0/N}$. For brevity, we do not specify the exact value of $c$ in the formal definition, but explicitly show in the Appendix that choosing $c\approx 0.02$ suffices (See Sec.~\ref{sec-app:ge-fee-Hs}). 

It is well known that solutions to NP-Hard problems can be encoded into the ground states of the problem Hamiltonians we consider. More concretely, consider the $2$-local classical Ising model Hamiltonian $H_\sigma$ of $n$-qubits as follows
\begin{equation}\label{eq:Ising_Ham}
H_{\sigma} = \sum_{\braket{i,j}} J_{ij} \sigma_z^{i} \sigma_z^{j} + \sum_{j=1}^{n}h_j\sigma_z^{j}
\end{equation}
where $J_{ij},~h_j \in \{-m,~-m+1,~\ldots~,~0,~1, \ldots,~m-1,~m\}$ for some constant positive integer $m\in \Theta(1)$. There are $M$ distinct eigenvalues of $H_{\sigma}$, all of which are integers and $M\in \poly(n)$. This is the quantum version of Ising spin glass Hamiltonians, and solutions to several \NP-hard (or \NP-complete) problems can be encoded in the ground states of $H_{\sigma}$. In fact, for some \NP-complete problems such as Quadratic Unconstrained Binary Optimization (QUBO) or MaxCut, the encoding is straightforward, with very little overhead on $n$. For details, we refer the readers to Refs.~\cite{barahona1982computational, lucas2014ising}. Note that the normalized version of $H_{\sigma}$ satisfies the conditions outlined in Definition \ref{def:prob-Ham}. In particular, it has a spectral gap $\Delta\geq 1/\poly(n)$, which, combined with the lower bound on the quantity $A_2$, ensures that the condition in Eq.~\eqref{eq:spectral-condition} holds. Nevertheless, we shall keep our analysis general as it works for a broader class of local spin Hamiltonians. We shall invoke particular cases (such as the classical Ising Hamiltonian) only when necessary.

The AQO algorithm we consider starts from $\ket{\psi_0}$ (the ground state of the initial Hamiltonian) and adiabatically evolves into a quantum state that has a fidelity of at least $1-\varepsilon$ with the ground states of $H_z$. 
\begin{figure}[h!]
    \centering
    \includegraphics[width=.9\linewidth]{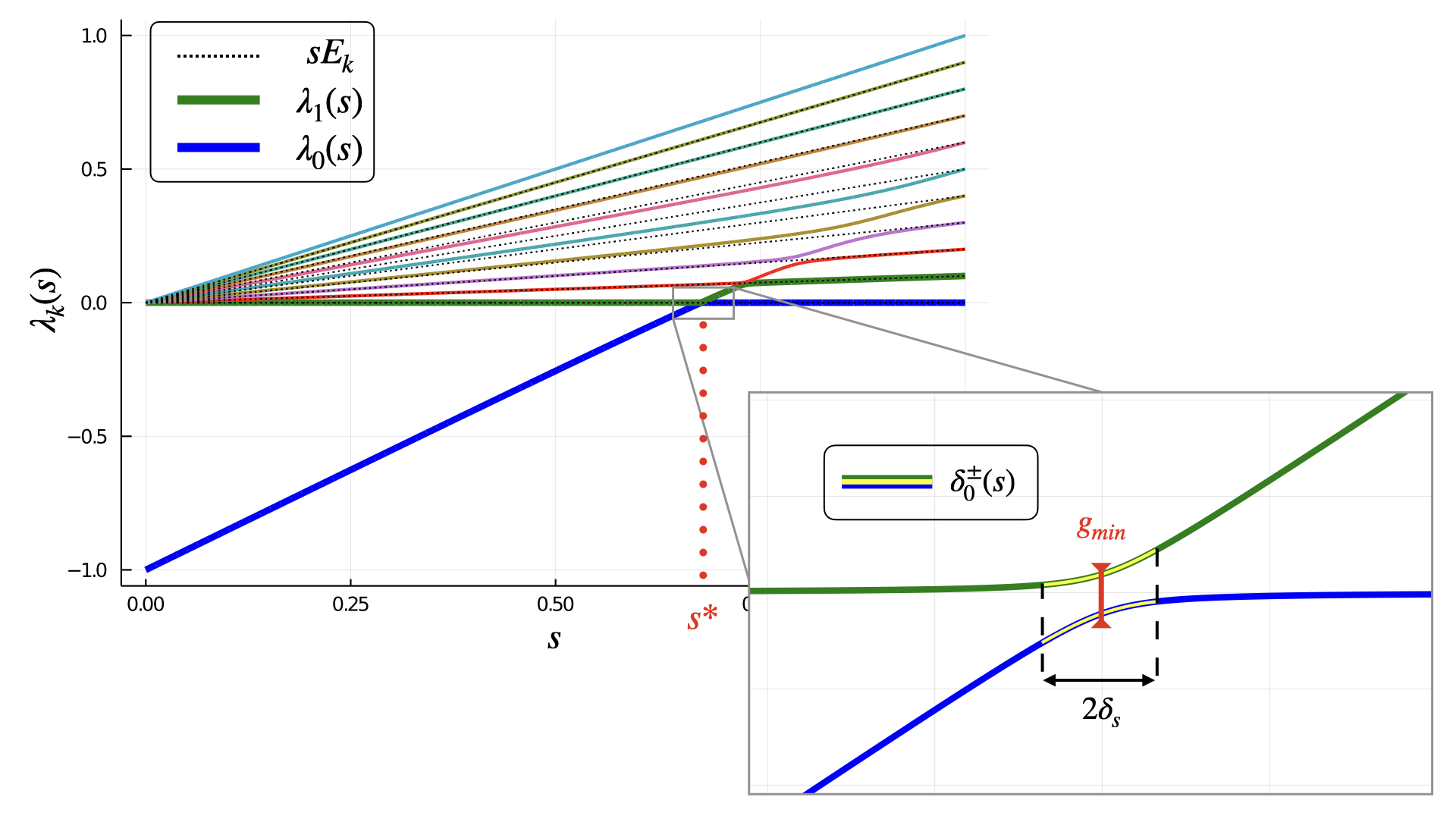}
    \caption{Spectrum of the adiabatic Hamiltonian of $n=20$ qubits, denoted by $H(s)=-(1-s)\ket{\psi_0}\bra{\psi_0}+sH_z$, where the schedule $s\in [0,1]$. The initial Hamiltonian is a one-dimensional projector such that $\ket{\psi_0}$ is an equal superposition of all computational basis states of $n$ qubits. The final Hamiltonian $H_z$ is diagonal in the computational basis such that there are $11$ distinct eigenenergies $E_k$, equally spaced between $0$ and $1$. Each eigenvalue $E_k$ is $d_k$-degenerate, where the degeneracies are distributed according to a Gaussian probability distribution.  The vertical dotted line (red) denotes the position of the avoided crossing between the ground and the first excited state. The inset figure isolates the two lowest eigenstates of $H(s)$ from the rest of the spectrum and zooms into the region of the avoided crossing. As described in the article, the region of the avoided crossing corresponds to a window of width $\delta_s$ on either side of $s^*$. The ground and the first excited states are given by $\lambda_0(s)=sE_0+\delta^+_0(s)$, and $\lambda_1(s)=sE_0+\delta^-_0(s)$, with the gap between these eigenstates scaling as $g_{\min}$, which is also the minimum spectral gap of $H(s)$.}
    \label{fig:spectrum-H(s)}
\end{figure}
Much like the adiabatic version of Grover's algorithm, having a one-dimensional projector as the initial Hamiltonian ensures that the spectrum of the $H(s)$ has only a single avoided crossing between the ground and the first excited state. However, the position of this avoided crossing is non-trivial in our case and depends on the spectrum of the problem Hamiltonian $H_z$. Moreover, the spectrum is complicated by the presence of avoided crossings between the higher excited states, which makes it challenging to bound the spectral gap of $H(s)$ (denoted by $g(s)$) throughout the adiabatic evolution.  As mentioned previously, this is precisely why, in such settings, generic provable speedups over classical unstructured search have been absent.

In Sec.~\ref{subsec:poac}, we provide an approximation of the position of the avoided crossing between the ground and the first excited states. This involves (a) finding the point where the spectral gap of $H(s)$ is minimum, (b) identifying a narrow window of $s$, in which the spectral gap scales similarly to the minimum gap, and (c) proving that the spectral gap is at least as large outside this window. Indeed, we prove that the position of the avoided crossing is well approximated by 
\begin{equation}
\label{eq:position-ac}
s^*=\dfrac{A_1}{A_1+1}.
\end{equation}
The spectral gap of $H(s)$ (denoted by $g(s)$) remains close to the minimum gap within a small window $\delta_s$ around $s^*$ (See inset Fig.~\ref{fig:spectrum-H(s)}). That is, the actual position of the avoided crossing is within the interval $\mathcal{I}_{s^*}=[s^*-\delta_s, s^*+\delta_s]$, where
\begin{equation}
\label{eq:interval-ac}
\delta_s=\dfrac{2}{(A_1+1)^2}\sqrt{\dfrac{d_0A_2}{N}}.
\end{equation}
Thus, for any $s\in\mathcal{I}_{s^{*}}$, the spectral gap of the Hamiltonian $H(s)$, satisfies $g(s)=O(g_{\min})$, where
\begin{equation}
 g_{\min}=\dfrac{2A_1}{A_1+1}\sqrt{\dfrac{d_0}{A_2 N}},
\end{equation}
is the minimum spectral gap of $H(s)$. 

Outside this window, we define the intervals: (i) $\mathcal{I}_{s^{\leftarrow}}=[0,~s^*-\delta_s)$, and (ii) $\mathcal{I}_{s^{\rightarrow}}=(s^*+\delta_s,~1]$. We use different techniques for each of these regions to obtain lower bounds on the spectral gap of $H(s)$. To the left of the avoided crossing, i.e.\ for any $s\in \mathcal{I}_{s^{\leftarrow}}$, we use the fact that for any ansatz $\ket{\phi}$ (of unit norm), the variational principle ensures that the ground energy $\lambda_0(s)\leq \braket{\phi|H(s)|\phi}$. We come up with a non-trivial ansatz such that this quantity is a tight upper bound on the ground energy (See Fig.~\ref{fig:lower-bound-gap}). Furthermore, we show that $sE_0$ is a good lower bound for the energy of the first excited state. These two bounds combined allow us to obtain a tight lower bound on $g(s)$ in this whole region. This enables us to prove that
\begin{equation}
g(s)\geq \dfrac{A_1}{A_2}\dfrac{s^*-s}{1-s^*},
\end{equation}
for any $s\in \mathcal{I}_{s^{\leftarrow}}$.

\begin{figure}[h!]
    \centering
    \includegraphics[width=0.9\linewidth]{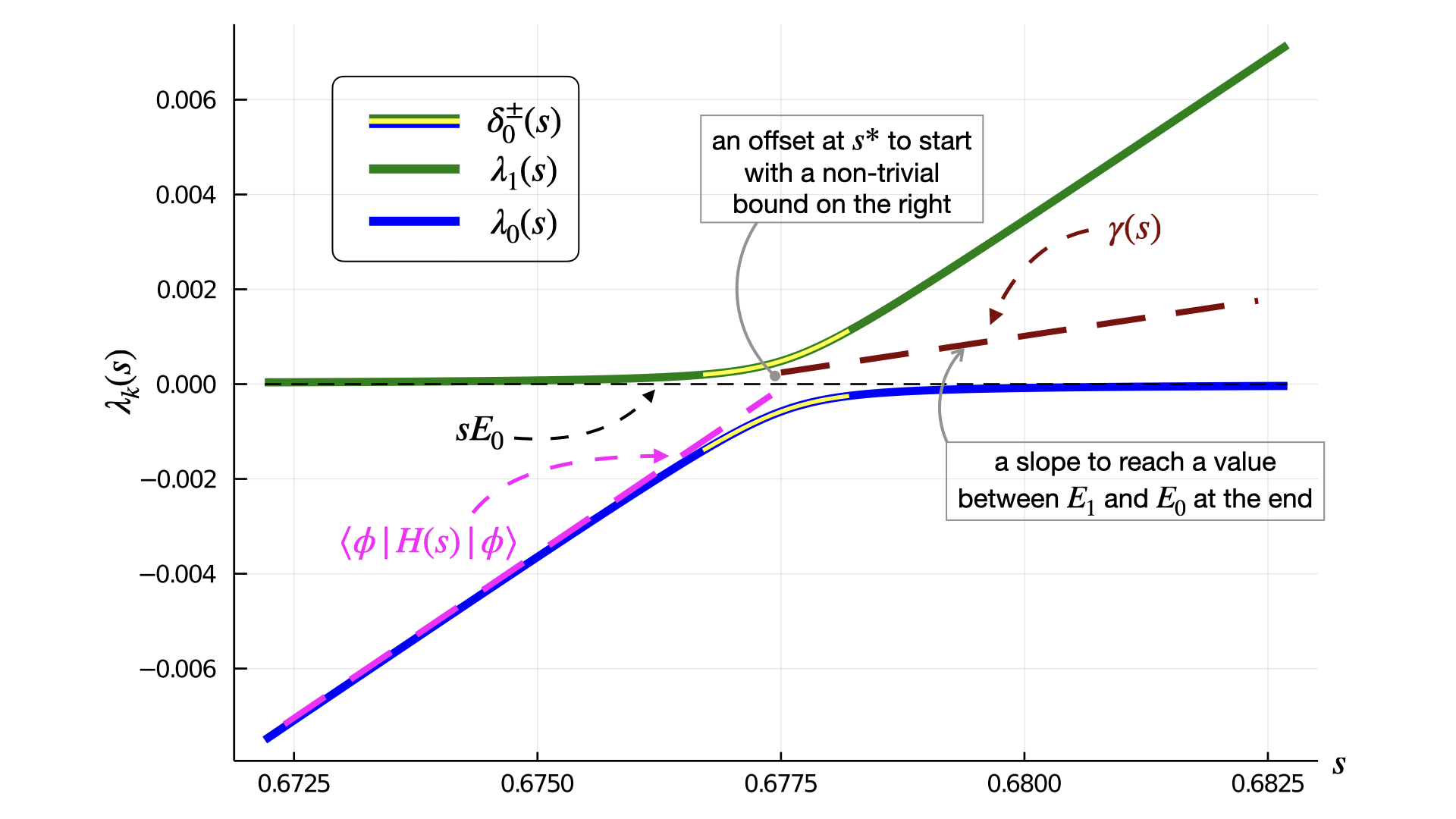}
    \caption{The lower bound on the spectral gap of the adiabatic Hamiltonian $H(s)$ (denoted by $g(s)$) to the left and the right of the avoided crossing. We employ different techniques to bound the gap in each of these regions.\newline \textit{Left of the avoided crossing region:} We obtain an upper bound on the ground energy $\lambda_0(s)$ and a lower bound on the energy of the first excited state, $\lambda_1(s)$. We consider an ansatz $\ket{\phi}$ and use the variational principle to upper bound the ground energy $\lambda_0(s)\leq \braket{\phi|H(s)|\phi}$ (depicted by pink dotted lines). Additionally, $\lambda_1(s)$ is lower bounded by $sE_0$, which is depicted by a black horizontal dotted line. \newline\textit{Right of the avoided crossing region:~} We consider a straight line (dark red dotted lines) from $s=s^*$ to some value between the two lowest eigenvalues of the problem Hamiltonian ($E_0$ and $E_1$, respectively). The line $\gamma(s)$ is offset by a non-trivial amount at $s=s^*$ and its slope is carefully adjusted to obtain a tight bound on $g(s)$ in this region.}   
    \label{fig:lower-bound-gap}
\end{figure}

The same technique does not yield a tight bound on $g(s)$ to the right of the avoided crossing, as the spectrum is significantly more complicated (see Fig.~\ref{fig:spectrum-H(s)}). Indeed, it is considerably more challenging to obtain a tight bound on $g(s)$ here. We consider a line $\gamma(s)$ that lies between the lowest and the second lowest eigenvalues of $H(s)$ and show that the spectrum is at least a certain distance from this line by considering the resolvent of $H(s)$ (see Eq.~\eqref{eq:resolvent-definition}). More precisely, as shown in Fig.~\ref{fig:lower-bound-gap},  $\gamma(s)$ is offset by a non-trivial amount from $sE_0$ (the offset is precisely the quantity $kg_{\min}$, for some constant $k<1$) at $s=s^*$ and linearly interpolates between this point and an appropriately chosen value between $E_0$ and $E_1$ (say $E_0+a$, with $a<\Delta$), for $s=1$. That is, $\gamma(s)=s E_0+\beta(s)$, where
$$
\beta(s)=a\left(\dfrac{s-s_0}{1-s_0}\right),
$$
and 
$$
s_0=s^*-kg_{\min}\dfrac{1-s^*}{a - k~g_{\min}}.
$$
 Here $a$ determines the slope of the line $\gamma(s)$ and $\beta(1)=a< \Delta$. Both $a$ and $k$ are carefully tuned to obtain a tight lower bound on $g(s)$ in this region. Indeed, for any $s\geq s^*$, we consider $R_{H(s)}(\gamma)$, the resolvent of $H(s)$ with respect to $\gamma(s)$, and estimate an upper bound on $\norm{R_{H(s)}}$ using the Sherman-Morrison formula (See Eq.~\eqref{eq:sherman-morrison}). The distance between the line $\gamma(s)$ and the spectrum of $H(s)$ is then given by the quantity $\norm{R_{H(s)}}^{-1}$. Finally, we leverage the fact that the spectral gap $g(s)$ is at least twice the minimum distance of $\gamma(s)$ from the spectrum of $H(s)$, i.e.
$$
g(s)\geq 2/\norm{R_{H(s)}(\gamma)}.
$$ 
This allows us to bound the gap to the right of the avoided crossing as 
\begin{equation}
g(s)\geq \dfrac{2\beta(s)}{1+f(s)},
\end{equation}
for some function $f(s)$. We rigorously establish that the choices $k=1/4$, and $a=4k^2\Delta/3$ ensure that monotonically $f$  is a decreasing function in the interval $s\in [s^*, 1]$, such that its maximum value $f(s^*)=\Theta(1)$. In fact, this enables us to obtain a simple bound on the gap for any $s\in [s^*,1]$, given by
\begin{equation}
\label{eq:gap-right-ac}
g(s)\geq \dfrac{\Delta}{30}\left(\dfrac{s-s_0}{1-s_0}\right).
\end{equation}
Since our choice of $\beta$ ensures that $\beta(s^*)\geq k g_{\min}$, the gap can be lower bounded as $g(s)\geq O(g_{\min})$ in the entire interval of $s\in [s^*,1]$, including in $\mathcal{I}_{s^{\rightarrow}}$. 
\begin{figure}[h!]
    \centering
    \includegraphics[trim={7cm 0 0 0},clip, width=0.8\linewidth]{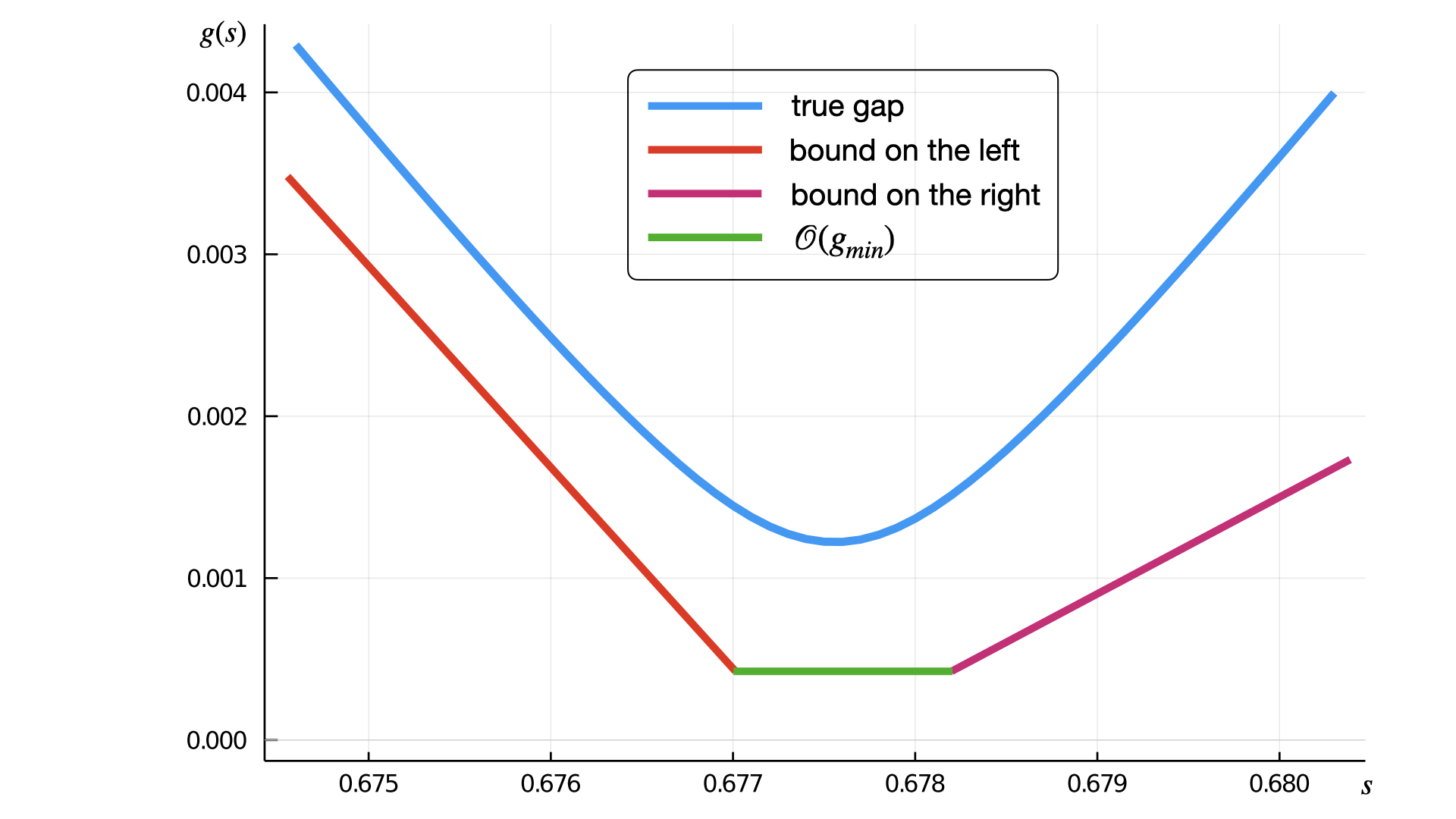}
    \caption{Comparison of the actual gap of the adiabatic Hamiltonian $H(s)$ (denoted by $g(s)$) in the article and the lower bounds we obtained for a subset of the interval $s\in [0,1]$. The line above (blue) denotes the true value of $g(s)$ for a $20$-qubit adiabatic Hamiltonian $H(s)=-(1-s)\ket{\psi_0}\bra{\psi_0}+sH_z$ with $s\in [0,1]$. Here, $\ket{\psi_0}$ is the equal superposition of all $20$-qubit computational basis states, while the problem Hamiltonian $H_z$ is diagonal in the computational basis. $H_z$ has $11$ distinct eigenvalues $E_k$ such each $E_k$ is $d_k$-degenerate. The degeneracies are distributed according to a Gaussian probability distribution. As described in the article, we divide $s\in [0,1]$ into three intervals and obtain bounds on the gap using different techniques in each interval. The line below corresponds to our lower bound on $g(s)$: to the left of the avoided crossing region (in red), the region of the avoided crossing (in green), and to the right of the avoided crossing (dark pink).}    
\label{fig:true-gap-vs-lower-bound}
\end{figure}
From these bounds, we are able to prove that $g(s)$ is minimum in the vicinity of the avoided crossing (i.e.\ for $s\in\mathcal{I}_{s^*}$) and larger outside this window. A comparison between the true value of $g(s)$ and our lower bounds in each of the three intervals is shown in Fig.~\ref{fig:true-gap-vs-lower-bound}. Obtaining tight bounds on the spectral gap for any $s$, (i) allows us to construct the optimal local schedule, and (ii) apply the adiabatic theorem. It is well known, following Refs.~\cite{vandam2001powerful, roland2004quantum} that a local schedule is necessary for obtaining optimal running times in adiabatic algorithms. Such a schedule requires information about the scaling of the gap as a function of $s$, ensuring that the evolution is slower in regions where the gap is small. In this regard, we present an extension of the recent result of Ref.~\cite{cunningham2024eigenpath} on quantum phase randomization to the continuous-time (adiabatic)  setting. This allows us to develop a simplified version of the adiabatic theorem that is quite general and requires minimal assumptions on $H(s)$: it holds for any bounded, twice differentiable Hamiltonian with a known lower bound on its gap $g(s)$. We are able to obtain a generic expression for the running time of any adiabatic algorithm under a local adaptive schedule whose derivative scales with $g(s)$. In Sec.~\ref{subsec:schedule-run-time}, we apply these bounds in the context of AQO, in conjunction with the bounds obtained for $g(s)$, to obtain a closed-form expression on the running time $T$, which matches the lower bound of \cite{farhi2008fail} (up to a polylogarithmic overhead) for a broad class of Hamiltonians. We summarize our findings via the following theorem:
~\\
\begin{restatable}[Main result 1: Running time of AQO]{theorem}{mainresultone}
\label{thm:main-result-1}
Let $\varepsilon\in [0,1)$ and consider the adiabatic Hamiltonian 
$$
H(s)=-(1-s)\ket{\psi_0}\bra{\psi_0}+s H_z,
$$ 
such that $\ket{\psi_0}=\ket{+}^{\otimes n}$ and $H_z$ satisfies the conditions of Definition \ref{def:prob-Ham}.  Then, adiabatic quantum optimization prepares a quantum state that has a fidelity of at least $1-\varepsilon$ with an equal superposition of the ground states of $H_z$, given by
$$
\ket{v(1)}=\dfrac{1}{\sqrt{d_0}}\sum_{z\in \Omega_0}\ket{z},
$$
in time 
$$
T=O\left(\dfrac{1}{\varepsilon}\cdot\dfrac{\sqrt{A_2}}{A_1^2\Delta^2}\cdot \sqrt{\dfrac{2^n}{d_0}}\right).
$$
\end{restatable} 
~\\
The algorithmic running time in Theorem \ref{thm:main-result-1} is optimal up to polylog factors for any Ising Hamiltonian $H_{\sigma}$ with a spectral gap $\Delta>1/\poly(n)$ (as defined in Eq.~\eqref{eq:Ising_Ham}). This is because, for these Hamiltonians, both $A_1$ and $A_2$ are lower bounded by a constant and upper bounded by $O(\poly(n))$, leading to a running time of $\widetilde{O}(\sqrt{2^n/d_0})$, which is optimal in this setting.

We observe that this generic quantum speedup over brute-force search comes with caveats. Note that to run the AQO algorithm, we need to construct the appropriate local schedule (whose derivative scales with the gap $g(s)$). For this, we argue that prior knowledge of the position of the avoided crossing is required to an additive precision of $O(\delta_s)$. This, in turn, necessitates the estimation of the spectral parameter $A_1$ prior to the running of the adiabatic algorithm. Additionally, the time required to compute this quantity should be no more than the running time $T$ of the adiabatic algorithm. 

In Sec.~\ref{sec:hardness-avoided-crossing}, we rigorously investigate how difficult it is to compute the position of the avoided crossing (which boils down to estimating $A_1$). More precisely, suppose we have access to an adiabatic quantum optimizer and a classical computing device. Then, what is the computational hardness of exactly, and subsequently, approximately estimating $A_1$? We prove that it is \NP-hard to approximate $A_1$ even to within an additive precision of $1/\poly(n)$ (which is much larger than the desired accuracy) in Sec.~\ref{sec:np-hardness-A1}. To this end, we consider Boolean satisfiability (more precisely, 3-\SAT), a well-known \NP-complete problem. This problem asks if a given Boolean formula has a satisfying assignment, and solving it reduces to whether we can disambiguate between two promised thresholds of the ground energy of a 3-local Hamiltonian $H$. We are able to prove that it is possible to distinguish between these two thresholds of the ground energy by making only two calls to any classical algorithm that estimates $A_1$ to an additive accuracy of $1/\poly(n)$. Our proof involves modifying the Hamiltonian $H$ by coupling an extra spin and then querying the classical algorithm to estimate $A_1$ for the modified Hamiltonian. Formally, we state the following lemma, which we prove in Sec.~\ref{sec:np-hardness-A1}:
\begin{theorem}[Main Result 2: Hardness of approximately estimating $A_1$]
\label{thm:np-hardness-3sat}
Let $\varepsilon\in (0,1)$. Suppose there exists a classical procedure $\mathcal{C}_{\varepsilon}(\langle H \rangle)$ that accepts as input, the description of a Hamiltonian $H$ and outputs $\widetilde{A}_1(H)$ such that
$$
\left|\widetilde{A}_1(H)-A_1(H)\right|\leq \varepsilon.
$$
Then, it is possible to solve the $3$-\SAT\ problem by making only two calls to $\mathcal{C}_{\varepsilon}$, provided
$$
\varepsilon < \dfrac{1}{72}\cdot \dfrac{1}{n-1}.
$$
\end{theorem}
In addition to this, we prove in Sec.~\ref{sec:sharp-P-hardness-A1}, that estimating $A_1$ exactly (or nearly exactly, i.e.,\ to within an accuracy of $2^{-\poly(n)}$) is as hard as solving any problem in \sharpP, i.e.\ it is \sharpP-hard. For this, consider a classical procedure $\mathcal{C}$ that exactly estimates $A_1$, then we show that it is possible to extract the degeneracies $d_k$ in the spectrum of the Ising Hamiltonian $H_{\sigma}$ using only $O(\poly(n))$ calls to $\mathcal{C}$. Once again, we modify $H_{\sigma}$ by adding an extra spin qubit of a certain local energy and use $\mathcal{C}$ to estimate $A_1$ for this modified Hamiltonian. By varying this local energy term of the additional spin, we estimate $A_1$ for $O(\poly(n))$ different values, requiring only $O(\poly(n))$ calls to $\mathcal{C}$. By using standard polynomial interpolation techniques, we are able to construct a polynomial from which the degeneracies $d_k$  of $H_{\sigma}$ can be extracted exactly. 

Recall that solutions to \NP-hard problems can be encoded in the ground states of $H_{\sigma}$, including 3-\SAT. The counting version of this problem (\#3-\SAT) asks the number of satisfying assignments of a given Boolean formula and is \sharpP-complete. This is equivalent to extracting the degeneracy of the ground state $d_0$, of $H_{\sigma}$. Thus, our reduction proves that computing $A_1$ exactly is \sharpP-hard. This result is robust to sufficiently small errors in the approximation of $A_1$. That is, the problem remains \sharpP-hard as long as $\mathcal{C}$ estimates $A_1$ to within an accuracy of $\varepsilon\in O(2^{-\poly(n)})$. It can also be extended to the setting where $\mathcal{C}$ is probabilistic. We summarize the results via the following theorem which we prove in Sec.~\ref{sec:sharp-P-hardness-A1}:

\begin{theorem}[Main Result 3: Hardness of (nearly) exactly estimating $A_1$]
\label{thm:main-result-3}
Suppose there exists a classical algorithm $\mathcal{C}(\langle H\rangle)$ which accepts as input the description of an $n$-qubit Hamiltonian $H$, and outputs $A_1$ exactly, or with an accuracy of $\varepsilon$ such that $\varepsilon\in O(2^{-\poly(n)})$. Furthermore, suppose $H_{\sigma}$ be the $n$-qubit Ising Hamiltonian (with appropriate parameter ranges, as defined in Eq.~\eqref{eq:Ising_Ham}), such that its eigenenergies are $E_0<E_1<\cdots E_{k}$, where $0\leq k \leq M-1$, with known gaps $\Delta_{k}=E_k-E_0$. Then for all $k\in [0, M-1]$, it is possible to estimate the degeneracy $d_k$ of energy eigenvalue $E_k$ by making $O(\poly(n))$ calls to $\mathcal{C}$. 
\end{theorem}

Overall, our results demonstrate that the optimality of AQO is contingent on solving a computationally hard problem a priori. As discussed earlier, this points to a fundamental limitation of the framework of AQO, which is absent in the circuit model, and we leave open the question of whether this can be circumvented without requiring access to a digital quantum computer. 

\subsection{Related Work}
\label{subsec:prior-work}
As discussed previously, the question we ask in this paper is whether a generic quadratic speedup over classical brute force search can be obtained to solve computationally hard problems using AQO. For this, we consider $H(s)$ to be a linear interpolation between a one-dimensional projector and a problem Hamiltonian that is diagonal in the standard basis (such as a classical Ising spin Hamiltonian) with a large enough spectral gap.  In Ref.~\cite{horvat2006exponential}, using a combination of analytical and heuristic arguments, the authors showed that the minimum spectral gap of $H(s)$ is exponentially small for the 3-\SAT\ problem (it scales roughly as $\sqrt{d_0/2^n}$, where $d_0$ is the number of solutions), and also provided an expression for the position of the avoided crossing. However, a rigorous analysis of the algorithmic runtime was absent.  Hen \cite{hen2014continuous} showed that unstructured adiabatic search can achieve a quadratic speedup for a particular random classical problem Hamiltonian, with a specific distribution of energies. However, the analysis therein is simplified by the specific choice of the problem Hamiltonian which ensures that the position of the avoided crossing does not depend on its spectrum. 

The behavior of the spectrum of interpolated Hamiltonians (such as the $H(s)$ we consider) in the vicinity of an avoided crossing has been studied very soon after the first version of the adiabatic theorem \cite{born1928beweis}. We refer the readers to Ref.~\cite{arthurthesis} for an exposition of these early results. However, such techniques do not yield tight bounds for the gap, away from the avoided crossing. In our work, we manage to obtain tight bounds on the spectral gap for any $s\in [0,1]$ by using a variety of techniques, namely (i)~the variational principle, (ii)~methods used in the context of continuous-time quantum walks \cite{childs2004spatial, chakraborty2016spatial, chakraborty2020optimality}, and (iii) the Sherman-Morrison formula \cite{sherman_morrison}. We also develop a simplified version of the adiabatic theorem, which requires fewer assumptions than the prior art \cite{jansen2007bounds}. This extends the recent work of Ref.~\cite{cunningham2024eigenpath} to the continuous-time setting. Our general results can be used to obtain the time after which adiabatically evolving a state under Hamiltonian $H(s)$ with a local schedule results in a state having a high fidelity with the ground state of $H(1)$. These methods are more broadly applicable, and the details of the same will appear elsewhere \cite{cunningham2025forthcoming}. Here, we obtain upper bounds on the running time of AQO by applying the aforementioned results.

\section{Optimal adiabatic quantum optimization algorithm}\label{sec:poac-and-schedule}
In this Section, we rigorously derive the bounds on the spectral gap of the adiabatic Hamiltonian $H(s)$ and present the overall algorithm in detail. We begin by showing that it suffices to consider $H(s)$ restricted to an $M$-dimensional symmetric subspace (which it leaves invariant) and also derive expressions for the relevant eigenvalues and eigenstates.  

\subsection{Spectrum of the adiabatic Hamiltonian}
We begin by briefly recapping some of the details on the adiabatic Hamiltonian and the problem Hamiltonians we consider, as discussed in Sec.~\ref{subsec:prelim}. The overall adiabatic Hamiltonian we have is given by
$$
H(s)=-(1-s)\ket{\psi_0}\bra{\psi_0} + s H_z,
$$
where $\ket{\psi_0}$ is the equal superposition of all $N=2^n$
computational basis states, and $H_z$ is the problem Hamiltonian, satisfying Definition~\ref{def:prob-Ham}. Recall from Definition \ref{def:prob-Ham} that the $n$-qubit problem Hamiltonian $H_z$ is diagonal in the computational basis, with $M$ distinct eigenlevels with eigenvalues $0\leq E_0<E_1<\ldots<E_{M-1}\leq 1$, such that energy level with eigenvalue $E_k$ is $d_k$-degenerate. Formally, for $0\leq k\leq M-1$, we define a set of bit-strings
$$
\Omega_k=\left\{z~|~z\in\{0,1\}^n,\ H_z\ket{z}=E_k\ket{z}\right\},
$$
such that $|\Omega_k|=d_k$ is the degeneracy of eigenvalue $E_k$ and $\sum_{k} d_k = 2^n=N$.

Let us denote the $N$-dimensional Hilbert space corresponding to $H(s)$ as $\mathcal{H}$. In this section we show that the spectrum of $H(s)$ has two mutually orthogonal invariant subspaces, i.e.\ $\mathcal{H}=\mathcal{H}_S\oplus\mathcal{H}^{\perp}_S$ such that $\mathrm{dim}(\mathcal{H}_S)=M$ and $\mathrm{dim}(\mathcal{H}^{\perp}_S)=N-M$. We will first define an $M$ dimensional symmetric subspace $\mathcal{H}_S$, comprising of non-degenerate eigenlevels of $H_z$ as follows:
\begin{equation}\label{eq:Hs-def}
  \mathcal{H}_S=\text{span}\left\{\ket{k}|~0\leq k\leq M-1 \right\},
\end{equation}
where 
$\ket{k}=\frac{1}{\sqrt{d_k}}\sum_{z\in\Omega_k}\ket{z}$. Then, we can re-write
$$
H_{z}=\sum_{k=0}^{M-1} E_k \ket{k}\bra{k},
$$
as a spectral decomposition of $M$ non-degenerate eigenlevels. Moreover, the state
 $$\ket{\psi_0}=\sum_{k=0}^{M-1}\sqrt{\frac{d_k}{N}}\ket{k},$$
i.e.\ $\ket{\psi_0}\in \mathcal{H}_S$. Thus, the total Hamiltonian $H(s)$ leaves the subspace $\mathcal{H}_S$ invariant and can be written as
\begin{equation}
\label{eq:ham-symmetric-subspace}
\bar{H}(s)= -(1-s)\ket{\psi_0}\bra{\psi_0}+s\sum_{k=0}^{M-1} E_k\ket{k}\bra{k}.
\end{equation}

Let us now fix a basis for the complement of $\mathcal{H}_S$, which we refer to as $\mathcal{H}^{\perp}_S$. For every $k\in [0, M-1]$, let us order the bit strings in $\Omega_k$ as $z_k^{(\ell')}$, where $1\leq \ell' \leq |\Omega_k|=d_k$. Then, for all $k\in [0,M-1]$ and $\ell\in [0, d_k-1]$, define the Fourier basis
\begin{equation}
\ket{k^{(\ell)}}=\dfrac{1}{\sqrt{d_k}}\sum_{\ell'\in [d_k]} \exp\left[\dfrac{i2\pi\ell \ell'}{d_k}\right]\ket{z_k^{(\ell')}}.
\end{equation}
Note that for every $0\leq k\leq M-1$, we have $\ket{k^{(0)}}=\ket{k}\in \mathcal{H}_S$. Then, the $N-M$ dimensional subspace $\mathcal{H}^{\perp}_S$ is defined as 
\begin{equation}
\label{eq:complement-symm-subspace}
\mathcal{H}^{\perp}_S=\mathrm{span}\{\ket{k^{(\ell)}}|~0\leq k \leq M-1,~1\leq \ell \leq d_k-1\}.
\end{equation}
We now turn our attention to the eigenvalues of $H(s)$. Since $H(s)$ is the sum of a rank one projector and a Hermitian matrix, its eigenvalues, $\lambda(s)$, can be succinctly expressed \cite{golub1973modified}. We have the following lemma:   
\begin{lemma}
\label{lem:spectrum-H(s)}
Suppose $H(s)$ is the adiabatic Hamiltonian defined in Eq.~\eqref{eq:general-H(s)}. Then $\lambda(s)$ is an eigenvalue of $H(s)$ if and only if, either $\lambda(s)=s E_k$ or
$$
\dfrac{1}{1-s}=\dfrac{1}{N}\sum_{k=0}^{M-1}\dfrac{d_k}{sE_k-\lambda(s)}.
$$
\end{lemma}

\begin{proof}
Consider the $N-M$ dimensional subspace $\mathcal{H}^{\perp}_{S}$, then for any $\ket{k^{(\ell)}}$ we have
$$
H(s)\ket{k^{(\ell)}}=-(1-s)\ket{\psi_0}\underbrace{\braket{\psi_0|k^{(\ell)}}}_{=0} + s H_z\ket{k^{(\ell)}}=sE_k.
$$
So the subspace $\mathcal{H}^{\perp}_S$ is spanned by $N-M$ eigenstates of $H(s)$. 

Now consider the $M$ dimensional symmetric subspace $\mathcal{H}_S$. As discussed previously, this space is spanned by the states $\ket{k}$, where $k\in [0, M-1]$. For now, let us denote by $\bar{H}(s)$, the Hamiltonian $H(s)$ restricted to this invariant subspace, as defined in Eq.~\eqref{eq:ham-symmetric-subspace}. We begin by finding the conditions for $\ket{\psi}$ to be an eigenstate of $\bar{H}(s)$ with eigenvalue $\lambda(s)$. That is,
\begin{equation}
\label{eq:eigenstate-general-adiabatic-Ham}
\ket{\psi} = \sum_{k=0}^{M-1} \alpha_k \ket{k},
\end{equation}
where $\sum_{k}|\alpha_k|^2=1$. Then, as
\begin{equation}
H_{z} \ket{k} = E_k \ket{k},
\end{equation}
for all $0\leq k\leq M-1$, we have,
\begin{equation}
\bar{H}(s)\ket{\psi}=s \sum_{k\in[M]} E_k \alpha_k\ket{k} - (1-s)\gamma\ket{\psi_0} = \lambda \ket{\psi},
\end{equation}
where, we have assumed $\gamma = \braket{\psi_0|\psi}$. Now, expressing $\ket{\psi}$ as in Eq.~\eqref{eq:eigenstate-general-adiabatic-Ham}, and comparing each term, yields that for any $0\leq k\leq M-1$,
\begin{align}
& \lambda \alpha_k = sE_k\alpha_k -(1-s)\gamma\sqrt{\frac{d_k}{N}}                                                   \\
  \implies &\alpha_k= \gamma\dfrac{(1-s)\sqrt{d_k}}{\sqrt{N}\left(sE_k-\lambda\right)}
  \label{eq:expression-alpha-k}
\end{align}

Since, $\gamma=\braket{\psi_0|\psi}=(1/\sqrt{N})\sum_{k\in[M]}\alpha_k\sqrt{d_k}$, we can substitute the expression for $\alpha_k$ in Eq.~\eqref{eq:expression-alpha-k} to obtain
\begin{equation}\label{eq:summation-equality}
  \dfrac{1}{1-s} = \dfrac{1}{ N} \sum_{k=0}^{M-1}\dfrac{d_k}{sE_k - \lambda}
\end{equation}
This provides the conditions for $\lambda(s)$ to be an eigenvalue of $\bar{H}(s)$. Observe that the right hand side of this equation is a monotonically decreasing function of $\lambda$ within each interval $(
sE_{k-1}, s E_k)$, with $\lambda=sE_k$ being the poles. This ensures each of these intervals, including the interval $(-\infty, sE_0)$ has exactly one solution. Thus, overall there are $M$ solutions to this equation, and hence $M$ eigenvalues $\lambda(s)$, corresponding to eigenstates in $\mathcal{H}_S$.  
\end{proof}

Although Eq.~\eqref{eq:summation-equality} allows us to assign an interval to each eigenvalue, it is not straightforward to extract analytical bounds on the spectral gap of $H(s)$. Indeed, by observing that the two lowest eigenvalues satisfy $\lambda_0(s)\in \ ]-\infty, sE_0[$, and $\lambda_1(s) \in \ ]sE_0,sE_1[$, respectively, we only obtain the trivial bound of $g(s)\geq 0$. Obtaining tight bounds required a more fine-grained analysis, which we undertake in the subsequent sections.

Note that the $N-M$ eigenstates in $\mathcal{H}^{\perp}_S$ are not relevant for the adiabatic evolution. This is because the system is initialized in the ground state of $H_0$, $\ket{\psi_0}\in \mathcal{H}_S$, and consequently, throughout the adiabatic evolution, the dynamics is restricted to the symmetric subspace $\mathcal{H}_S$. Thus, for brevity, we shall henceforth refer to $\bar{H}(s)$, the restriction of the total adiabatic Hamiltonian to the symmetric subspace $\mathcal{H}_S$, as simply $H(s)$. So with this relabelling, we have
\begin{align}
  H(s)& =-(1-s)\ket{\psi_0}\bra{\psi_0}+s\sum_{k=0}^{M-1} E_k\ket{k}\bra{k}\\
             &= -(1-s)\ket{\psi_0}\bra{\psi_0}+s H_z.     
             \label{eq:def-H-sigma-bar}
\end{align}
Thus effectively, the problem Hamiltonian $H_z$, restricted to this subspace has $M$ distinct eigenvalues, $0\leq E_0< E_1<\cdots E_{M-1}< 1$. We shall refer to the gap between the two lowest non-degenerate eigenvalues as the spectral gap of $H_z$. That is, $\Delta=E_1-E_0$. Furthermore, the spectral parameters $A_p$ for $p\in\mathbb{N}$ are defined as
\begin{equation}
\label{eq:spectral-parameter-Ham}
A_p=\dfrac{1}{N}\sum_{k=1}^{M-1} \dfrac{d_k}{(E_k-E_0)^p},
\end{equation}
where recall that $d_k$ is the degeneracy of eigenvalue $E_k$. In the next section, we identify the position of the avoided crossing in the spectrum of $H(s)$ and obtain useful bounds on the spectral gap as a function of $s$, which will help us in (i) constructing the optimal adiabatic schedule and (ii) quantifying the running time of the resulting adiabatic algorithm.

\subsection{Finding the Position of Avoided Crossing}\label{subsec:poac}
We will divide the interval $[0,1]$, which is the range of the schedule function $s$, into three regions. For this define
\begin{equation}
s^*=\dfrac{A_1}{A_1+1},
\end{equation}
and 
\begin{equation}
\delta_s=\dfrac{2}{(A_1+1)^2}\sqrt{\dfrac{d_0 A_2}{N}}.
\end{equation}
Then $[0,1]=\mathcal{I}_{s^{\leftarrow}}\bigcup \mathcal{I}_{s^{*}}\bigcup \mathcal{I}_{s^{\rightarrow}}$, where
\begin{align}
\mathcal{I}_{s^{\leftarrow}}=\left[0,s^*-\delta_s\right),~
\mathcal{I}_{s^{*}}=\left[s^*-\delta_s, s^*+\delta_s\right],~\mathrm{and~}
\mathcal{I}_{s^{\rightarrow}}&=\left(s^*+\delta_s, 1\right]
\end{align}
We will first show that the spectral gap $g(s)$ is minimum in the interval $\mathcal{I}_{s^*}$. Consequently, we will look for the two smallest solutions of Eq.~\eqref{eq:summation-equality}, of the form $\lambda(s)=sE_0+\delta_{\pm}(s)$. Substituting this in Eq.~\eqref{eq:summation-equality}, we obtain 
\begin{equation}\label{eq:delta-s-full}
-\frac{d_0}{N\delta_\pm} + \frac{1}{N} \sum_{k=1}^{M-1} \frac{d_k}{s(E_k - E_0) - \delta_\pm}=\dfrac{1}{1-s}
\end{equation}
We prove that within a symmetric interval around $s^*$, the two smallest solutions to Eq.~\eqref{eq:delta-s-full} ($\delta_+$ and $\delta_-$ respectively) can be obtained. Note that this also allows us to bound the spectral gap of $H(s)$ in this window, as $g(s)=\delta_+(s)-\delta_-(s)$. More precisely, consider the interval
\begin{equation}
\label{eq:interval-around-ac}
\mathcal{I}_{s^*} = \left[s^*-\delta_s,~s^*+\delta_s\right].
\end{equation}
Then, we prove via Lemma \ref{lem:validity-of-approximations}, borrowing techniques from Ref.~\cite{chakraborty2020optimality}, that for any $s\in\mathcal{I}_{s^*}$, the two smallest solutions of Eq.~\eqref{eq:delta-s-full} are given by
\begin{align}
&\delta_+(s)\in\left((1-\eta)\delta^{+}_{0}(s),~(1+\eta)\delta^{+}_{0}(s)\right),~\text{ and }\\
&\delta_-(s)\in\left((1+\eta)\delta^{-}_{0}(s),~(1-\eta)\delta^{-}_{0}(s)\right),
\end{align}
where, $\eta$ is a small constant ($\ll 1$), and
\begin{equation}
\label{eq:delta_0}
\delta^{\pm}_{0}(s)=\dfrac{s(A_1+1)}{2A_2(1-s)}\left[~s-\dfrac{A_1}{A_1+1}\right.\left.\pm\sqrt{\left(\dfrac{A_1}{A_1+1}-s\right)^2+\dfrac{4A_2d_0}{N(A_1+1)^2}(1-s)^2}~\right]. 
\end{equation}
This immediately implies for $s\in \mathcal{I}_{s^*}$, the spectral gap of $H(s)$, which is $g(s)=\delta_+(s)-\delta_-(s)$, is simply $\delta^{+}_0(s)-\delta^{-}_0(s)$, up to a relative error of $2\eta$. Thus,
$$
g(s)=\dfrac{s(A_1+1)}{A_2(1-s)}\sqrt{\left(\dfrac{A_1}{A_1+1}-s\right)^2+\dfrac{4 A_2d_0}{N(A_1+1)^2}(1-s)^2}
$$
The minimum spectral gap in the interval $\mathcal{I}_{s^{*}}$ is at $s=s^*=A_1/(A_1+1)$. Indeed,
\begin{equation}
\label{eq:min-gap}
g_{\min}=\min_{s\in \mathcal{I}_{s^*}} g(s)=g(s^{*})\geq (1-2\eta)\cdot\dfrac{2 A_1}{1+A_1}\sqrt{\dfrac{d_0}{N A_2}}.
\end{equation}
Interestingly, for any $s\in \mathcal{I}_{s^*}$, we can obtain an upper bound on $g(s)$. We formally prove this via the following Lemma:
\begin{lemma}
\label{lem:spectral-gap-in-robustness-window}
Consider the positive constant $c$ from Definition \ref{def:prob-Ham}, and the adiabatic Hamiltonian $H(s)$ in Eq.~\eqref{eq:def-H-sigma-bar}. Furthermore, define
$$
\kappa'=\dfrac{1+2c}{1-2c}\sqrt{1+(1-2c)^2}.
$$ 
Then, for $s\in \mathcal{I}_{s^{*}}$, the spectral gap of $H(s)$ satisfies
\begin{equation}
g_{\min}\leq g(s)\leq \kappa'\cdot g_{\min}.
\end{equation}
\end{lemma}
\begin{proof}
We have already shown that $g_{\min}$ is the minimum gap of $H(s)$, for any $s\in\mathcal{I}_{s^*}$. Now,
\begin{align}
g(s)&= \dfrac{s(A_1+1)}{A_2(1-s)}\sqrt{\left(\dfrac{A_1}{A_1+1}-s\right)^2+\dfrac{4A_2d_0}{N(A_1+1)^2}(1-s)^2}
\end{align}
For the upper bound, we have that for any $s$ in $[s^*-\delta_s,s^*+\delta_s]$,
\begin{align}
g(s)&\leq  \dfrac{(A_1+1)}{A_2}\frac{s}{1-s}\sqrt{\delta_s^2+(1+A_1)^2\delta_s^2(1-s)^2}\\
&\leq s^*\dfrac{(A_1+1)^2}{A_2}\delta_s\frac{s}{s^*}\sqrt{\frac{1}{(1-s)^2(1+A_1)^2}+1} \\
&\leq s^*\dfrac{(A_1+1)^2}{A_2}\delta_s \left( 1+\frac{\delta_s}{s^*} \right )\sqrt{1+\frac{1}{(1-\frac{\delta_s}{1-s^*})^2}}.
\label{eq:g(s)-ub}
\end{align}
Now let us look at the terms in the RHS of Eq.~\eqref{eq:g(s)-ub}. We have 
$$
\dfrac{s^*(A_1+1)^2}{A_2}\delta_s=g_{\min}.
$$ 
So, from the spectral condition of the problem Hamiltonian in Definition \ref{def:prob-Ham}, i.e.\ 
$$
\frac{1}{\Delta}\sqrt{\frac{d_0}{A_2N}}<c
$$ and using $A_2\Delta \leq A_1$, we obtain
\begin{align*}
    \frac{\delta_s}{1-s^*} &= \frac{2}{1+A_1}\sqrt{\frac{d_0A_2}{N}} \\
    &= \frac{2A_2 \Delta}{1+A_1}\frac{1}{\Delta}\sqrt{\frac{d_0}{A_2N}} \\
    &\leq 2s^* c  \\
    &\leq 2c.
\end{align*}
Similarly, it is easy to see that $\delta_s/s^*\leq 2c$. Substituting these back into Eq.~\eqref{eq:g(s)-ub}, we have 
\begin{align}
     g(s)&\leq g_{\min}\cdot \left(\dfrac{1+2c}{1-2c}\sqrt{1+(1-2c)^2}\right)\\    
         &\leq g_{\min} \cdot \kappa'.
\label{eq:gap-scaling-robustness-window}
\end{align}
\end{proof}
Thus, in the entire interval $\mathcal{I}_{s^*}$, the gap $g(s)=\Theta(g_{\min})$ . We claim that the avoided crossing lies within this interval centered around $s^*$. To this end, we show that for any $s\in [0,1]\backslash\mathcal{I}_{s^*}$, the scaling of the spectral gap is at least as large as the scaling of the gap in $\mathcal{I}_{s^*}$.
~\\~\\ 
\textbf{Bounding the spectral gap to the left of the avoided crossing:} Define the interval $\mathcal{I}_{s^{\leftarrow}}=\left[0, s^*-\delta_s\right)$, which is to the left of the region of the avoided crossing. We obtain a tight lower bound on the spectral gap of $H(s)$ in this region, which is at least as large as the gap for any $s\in \mathcal{I}_{s^*}$. We come up with a non-trivial ansatz state $\ket{\phi}$ that allows us to leverage the variational principle and first obtain an upper bound on the ground energy as $\lambda_0(s)\leq \braket{\phi|H(s)|\phi}$. Then the lower bound on the gap $g(s)$ is obtained by using this bound, along with the fact that from Eq.~\eqref{eq:summation-equality}, the energy of the first excited state satisfies $\lambda_1(s)\geq s E_0$, in this region. Formally, we state the following lemma:

\begin{lemma}
Consider the adiabatic Hamiltonian $H(s)$ defined in Eq.~\eqref{eq:def-H-sigma-bar}. Then, for any $s\in \mathcal{I}_{s^{\leftarrow}}$, the spectral gap of $H(s)$ satisfies
\begin{equation}
g(s)\geq \dfrac{A_1 (A_1+1)}{A_2}(s^*-s)
\end{equation}
\end{lemma}
\begin{proof}
Suppose $\lambda_0(s),~\lambda_1(s)$ is the ground state energy and the first excited state energy of $H(s)$, respectively. Then, in order to obtain a lower bound on the gap $g(s)=\lambda_1(s)-\lambda_0(s)$, for any $s\in \mathcal{I}_{s^{\leftarrow}}$, we make use of the variational principle. We can obtain an upper bound on $\lambda_0(s)$ by noting that from the variational principle, for any ansatz $\ket{\phi}$ of unit norm, $\lambda_0(s) \leq \langle \phi |H(s)|\phi \rangle$. Consider the ansatz 

\begin{equation}
  |\phi\rangle = \frac{1}{\sqrt{A_2N}}\sum_{k=1}^{M-1} \dfrac{\sqrt{d_k}}{E_k-E_0}|k\rangle. 
\end{equation} 
It is easy to verify that $\langle \phi|\phi \rangle=1$. Thus,
\begin{align*}
    \lambda_0(s)\leq \langle \phi |H(s)|\phi \rangle &= -(1-s)|\langle \psi_0|\phi \rangle|^2+s \langle \phi |(H_1-E_0+E_0)|\phi \rangle \\
    &= -(1-s)\left (\frac{A_1}{\sqrt{A_2}} \right)^2+sE_0+s\langle \phi |(H_1-E_0)|\phi \rangle \\
    &= -(1-s)\frac{A_1^2}{A_2}+sE_0+s\langle \phi |\sum_{k=1}^{M-1} (E_k-E_0) |k\rangle\langle k|\phi \rangle \\
    &= -(1-s)\frac{A_1^2}{A_2}+sE_0+s\sum_{k=1}^{M-1} (E_k-E_0) |\langle k|\phi \rangle|^2 \\
    &= -(1-s)\frac{A_1^2}{A_2}+sE_0+s\frac{A_1}{A_2} \\
    &=sE_0+\dfrac{A_1}{A_2}\left[s(1+A_1)-A_1\right]\\
    &= sE_0+\frac{A_1}{A_2}\frac{s-s^*}{1-s^*}
\end{align*}
We shall use this upper bound on $\lambda_0(s)$ to estimate the gap to the left of the avoided crossing, along with the fact that $\lambda_1(s)\geq sE_0$ (which follows from Eq.~\eqref{eq:summation-equality}). We have
\begin{align*}
    g(s)&=\lambda_1(s)-\lambda_0(s) \\
    &\geq \dfrac{A_1}{A_2}\dfrac{s^*-s}{1-s^*}
\end{align*}
\end{proof}

Now, for any $s\in \mathcal{I}_{s^{\leftarrow}}$, we have $s^*-s>\delta_s$ and hence,
\begin{align*}
g(s)&> \dfrac{A_1 (1+A_1)}{A_2}\delta_s\\
    & > \dfrac{2A_1 }{1+A_1}\sqrt{\dfrac{d_0}{A_2 N}}\\
    & > \Theta\left(g_{min}\right),
\end{align*}
which implies that the spectral gap is larger than the scaling of the gap in $\mathcal{I}_{s^*}$.
~\\~\\
\textbf{Bounding the spectral gap to the right of the avoided crossing:~}Obtaining a tight bound on the spectral gap to the right of the avoided crossing is considerably more challenging. This is because the spectrum is more complicated in this region (See Fig.~\ref{fig:spectrum-H(s)}), and we require a different technique to keep track of the gap. We consider a straight line $\gamma(s)$ from $sE_0+k~g_{\min}$ (at $s=s^*$) to some point $a$ between the two lowest eigenvalues of $H_z$ at $s=1$, where $k<1$ is a non-trivial constant and $a$ controls the slope of this line. The gap $g(s)$ is then lower bounded by twice the inverse of the norm of the resolvent of $H(s)$ with respect to $\gamma(s)$, denoted by $R_{H(s)}(\gamma)$. We obtain an upper bound on the norm of resolvent (equivalently, a lower bound on the gap) by the Sherman-Morrison formula \cite{sherman_morrison}. Additionally, by carefully tuning $k$ and $a$, we obtain a good enough lower bound on $g(s)$ for any $s\geq s^*$ (including the interval $\mathcal{I}_{s^{\rightarrow}}=(s^*+\delta_s, 1]$). We can also prove that  $g_{\min}$ is a lower bound on the spectral gap of $H(s)$ in this interval. Formally, we state the following lemma:
\begin{lemma}
Let $k=1/4$, $g_{\min}$ be the minimum spectral gap of $H(s)$ as defined in Eq.~\eqref{eq:min-gap}, and $\Delta$ be the spectral gap of the problem Hamiltonian $H_z$. Furthermore, define $a=4k^2\Delta/3$, and
$$
s_0=s^*-\dfrac{k~g_{\min}\left(1-s^*\right)}{a - k~ g_{\min}}.
$$ 
Then for all $s\geq s^*$, the spectral gap of $H(s)$ is lower bounded as 
\begin{equation}\label{eq:lower-bound-spectral-gap-right}
g(s)\geq \dfrac{\Delta}{30}\dfrac{s-s_0}{1-s_0}.
\end{equation}
\label{lemma:right-gap-lower-bound}
\end{lemma}
\begin{proof}
We choose $\gamma(s)$ to be a line that lies between the lowest and second lowest eigenvalues of $H(s)$. More precisely, $\gamma(s)=s E_0+\beta(s)$, where
$$
\beta(s)=a\left(\dfrac{s-s_0}{1-s_0}\right).
$$
We shall choose the appropriate values of $a$ and $k$ later. The spectral gap of $H(s)$ is at least twice the minimum distance of $\gamma(s)$ from the spectrum of $H(s)$. That is,
\begin{equation}
g(s)\ge\dfrac{2}{\norm{R_{H(s)}(\gamma)}},
\end{equation}
where $R_{H(s)}$ is the resolvent (See Eq.~\eqref{eq:resolvent-definition} for the definition). Thus, we simply need to obtain an upper bound on $\norm{R_{H(s)}(\gamma)}$. To this end, applying the Sherman-Morrison formula (See Eq.~\eqref{eq:sherman-morrison}), followed by the triangle inequality, we obtain
\begin{equation}
\label{eq:resolvent-bound-raw}
\norm{R_{H(s)}(\gamma)}\leq \norm{R_{sH_z}(\gamma)} + (1-s)\dfrac{\norm{R_{sH_z}(\gamma)\ketbra{\psi_0}{\psi_0}R_{sH_z}(\gamma)}}{1+(1-s)\langle\psi_0|R_{sH_z}(\gamma)|\psi_0\rangle}
\end{equation}

First observe that $\norm{R_{sH_z}(\gamma)}=\beta^{-1}$. Now, for the second term of Eq.~\eqref{eq:resolvent-bound-raw}, we calculate bounds on the numerator and denominator separately. For the denominator, we first expand according to the definition of the resolvent in Eq.~\eqref{eq:resolvent-definition}, and then consider its Taylor Series expansion in $\beta$ to get
\begin{equation}
1+(1-s)\langle\psi_0|R_{sH_z}(\gamma)|\psi_0\rangle=1+\dfrac{(1-s)d_0}{N\beta}-\dfrac{1-s}{sN}\sum_{k=1}^{M-1}\dfrac{d_k}{E_k-E_0}\sum_{\ell=0}^\infty\left(\dfrac{\beta}{s(E_k-E_0)}\right)^\ell.   
\end{equation}
For our choice of $\beta(s)$, we have $\beta(s)\le s(E_k-E_0)/2$ for all $k\in[1,M-1]$. This immediately gives us the following lower bound on the denominator
\begin{equation}
\label{eq:resolvent-den-lower-bound}
1+(1-s)\left<\psi_0|R_{sH_z}(\gamma)|\psi_0\right>\ge1+\dfrac{(1-s)d_0}{N\beta}-(1-s)\left(\dfrac{A_1}{s}+\dfrac{2A_2\beta}{s^2}\right).
\end{equation}

Similarly, for the numerator, we obtain the following upper bound:

\begin{equation}\label{eq:resolvent-num-upper-bound}
\norm{R_{sH_z}(\gamma)\ketbra{\psi_0}{\psi_0}R_{sH_z}(\gamma)}\le\dfrac{d_0}{N\beta^2}+\dfrac{4A_2}{s^2}
\end{equation}

Substituting the bounds of Eq.~\eqref{eq:resolvent-den-lower-bound},~\eqref{eq:resolvent-num-upper-bound} and the value of $\norm{R_{sH_z}(\gamma)}=\frac{1}{\beta}$ in the bound derived in Eq.~\eqref{eq:resolvent-bound-raw}, gives us
\begin{equation}
\norm{R_{H(s)}(\gamma)}\leq \dfrac{1}{\beta(s)}\left[1 +\dfrac{ 1 +\dfrac{4N \beta^2 A_2}{s^2 d_0}}{1+\dfrac{N \beta }{d_0} \left(\dfrac{s-s^*}{s(1-s)(1-s^*)}\right) -\dfrac{2 N \beta^2 A_2}{s^2 d_0}}\right]=\dfrac{1}{\beta(s)}\left[1+f(s)\right].
\end{equation}

From this, we obtain the required lower bound on the spectral gap of $H(s)$ for any $s\geq s^*$ (which includes $s\in\mathcal{I}_{s^{\rightarrow}}$). That is, 
\begin{align}
g(s)&\geq 2\norm{R_{H(s)}(\gamma)}^{-1}\\
	&\geq \dfrac{2\beta(s)}{1+f(s)}.
\end{align}
From lemma \ref{lem:monotonic-dec-f}, we conclude that $f(s)$ is a monotonically decreasing function in the interval $s\in [s^*, 1]$ by fixing 
$$a=\frac{4}{3}k^2 \Delta.$$

The bound on the gap becomes:
\begin{align}
    g(s) &\geq \frac{2\beta(s)}{1+\max_s f} \\
    &\geq a \frac{1-8k^2}{1+4k^2}\frac{s-s_0}{1-s_0} \\
    &\geq  \frac{4}{3} k^2 \frac{1-8k^2}{1+4k^2}\Delta \frac{s-s_0}{1-s_0}
\end{align}
The best $k$ that maximizes the prefactor happens for $k=\frac{1}{2}\sqrt{\sqrt{\frac{3}{2}}-1} \simeq 0.237$ and $a=\frac{1}{3}(5-2\sqrt{6})\simeq 0.03367 \geq 1/30$ (value reached for $k=1/4)$. It turns into the following bound on the gap:
\begin{align}
    g(s) \geq \frac{\Delta}{30}\frac{s-s_0}{1-s_0}
\end{align}
\end{proof}
Note that from Lemma \ref{lemma:right-gap-lower-bound} we have that $f(s)$ is maximum at $s=s^*$, and moreover, $
f(s^*)= \Theta(1)$. Thus, $g(s^*)\geq O(\beta(s^*))=O(g_{\min})$. From this, we conclude that the lower bound on the spectral gap given in Lemma \ref{lemma:right-gap-lower-bound} satisfies, $g(s)\geq O(g_{\min})$ for all $s\in [s^*, 1]$. This completes the picture vis-a-vis the spectral gap of $H(s)$: the gap is minimum in the vicinity of the avoided crossing ($s\in \mathcal{I}_{s^*}$) and is at least as large outside this interval ($s\in \mathcal{I}_{s^{\leftarrow}}$ and $s\in \mathcal{I}_{s^{\rightarrow}}$, respectively). The comparison between the true gap $g(s)$ and our estimated lower bound for each of the three intervals is illustrated in Fig.~\ref{fig:true-gap-vs-lower-bound} for a $20$-qubit Hamiltonian where $H_z$ is diagonal in the computational basis with uniformly spaced eigenenergy levels between $0$ and $1$. Each such eigenlevel is $d_k$ degenerate, such that the degeneracies are distributed according to a Gaussian distribution. The knowledge of the lower bound on the spectral gap of $H(s)$ in all the regions of its domain allows us to construct the optimal local adiabatic schedule and subsequently estimate the running time of the overall adiabatic algorithm. We will discuss this in the next section.

\subsection{Optimal adiabatic schedule and running time}
\label{subsec:schedule-run-time}

In order to obtain the running time of AQO, we first develop a simplified version of the adiabatic theorem, extending the recent results of Ref.~\cite{cunningham2024eigenpath} on phase randomization to the continuous-time (adiabatic) setting. Our results are quite generic: they hold for any parametrized Hamiltonian $H(s)$ that is twice differentiable, with a known lower bound on its spectral gap $g(s)$. We rigorously derive the time $T$ after which evolving an initial state under $H(s)$ with a local schedule has a fidelity of at least $1-\varepsilon$ with the ground state of the problem Hamiltonian. In particular, we consider a schedule relevant to our case, namely a local schedule whose derivative scales with the inverse of the instantaneous gap. The detailed proofs of our generic results on the adiabatic theorem have been derived in the Appendix (Sec.~\ref{sec-app:adiabatic-theorem}). In this section, we discuss the running time of AQO, which uses this version of the adiabatic theorem. 

In order to obtain the running time, all we need are the bounds obtained on the gap $g(s)$ in Sec.~\ref{subsec:poac}. For the adiabatic evolution, we make use of a local adaptive schedule that takes into account the instantaneous gap $g(s)$. This ensures that the system evolves rapidly in regions where the gap $g(s)$ is large and slows down sufficiently when $g(s)$ is small (in and around the vicinity of the avoided crossing). In particular, we make use of a schedule whose derivative scales with the gap $g(s)$, as has also been used in prior works \cite{roland2004quantum, vandam2001powerful}.  We state our result on the running time of AQO via the following Theorem (restated from Sec.~\ref{subsec:results}) while the detailed derivation can be found in the Appendix  (Sec.~\ref{sec-app:main-result-one-proof}): 

\mainresultone*

\textit{Proof:~}The proof is obtained by applying Theorem \ref{theorem:adaptiveRate} to AQO and is stated in detail in the Appendix (Sec.~\ref{sec-app:main-result-one-proof}). 
~\\
\qed 
~\\~\\
Now, when $H_z$ corresponds to the (normalized) classical Ising Hamiltonian (defined in Eq.~\eqref{eq:Ising_Ham}), we have $\Delta\geq 1/\poly(n)$ and $A_2\leq 1/\Delta^2 \leq \poly(n)$. This implies that the ground state of the problem Hamiltonian is prepared with a constant fidelity in a time $T=\widetilde{O}(\sqrt{2^n/d_0})$. This matches the lower bound of Ref.~\cite{farhi2008fail} (up to log factors) when the degeneracy $d_0$ does not scale with $2^n$.

Although these results prove a generic quadratic advantage over brute force search for AQO, there are caveats. The main disadvantage is the necessity of constructing a local schedule whose derivative scales with the instantaneous gap $g(s)$. We argue that for such a schedule, it is necessary to apriori predict the position of the avoided crossing $s^*$ with sufficient accuracy. Recall that the bounds on $g(s)$ for each of the three intervals are as follows: 
\begin{equation}
g(s)\geq 
\begin{cases}
\dfrac{A_1 (A_1+1)}{A_2}\left(s^*-s\right),& \qquad s\in \mathcal{I}_{s^{\leftarrow}}=\Big[0,~s^*-\delta_s\Big)\\~&~\\
g_{\min},& \qquad s\in\mathcal{I}_{s^{*}}=\Big[s^*-\delta_s,~ s^*\Big)\\~&~\\
\dfrac{\Delta}{30}\left(\dfrac{s-s^*}{1-s_0}\right)+\dfrac{\Delta}{30}\cdot\dfrac{1}{1-s_0}\cdot\dfrac{g_{\min}}{\Delta/3-g_{\min}},& \qquad s\in\mathcal{I}_{s^{\rightarrow}}=\Big[s^*,~ 1\Big]
\end{cases}
\end{equation}
Observe that in the vicinity of the avoided crossing (a window of size $\delta_s$ to either side of $s^*$), the rate of change in the schedule is constant. However, for this region, prior knowledge of $g_{\min}$ is an imperative necessity. The factor that dominates the algorithmic running time is, in fact, $2^{-n/2}$. Replacing the parameter $A_2$ with its (known) lower bound (which is a constant) suffices, while a guess of $d_0=1$ ensures that $g_{\min}$ is lower bounded by $2^{-n/2}$, quantity that is only smaller than $g_{\min}$ by a factor of $\poly(n)$.

On the other hand, away from the position of the avoided crossing, (in the intervals $\mathcal{I}_{s^{\leftarrow}}$ and $\mathcal{I}_{s^{\rightarrow}}$), the rate of change of the schedule is regulated by the distance from the avoided crossing ($s-s^*$ or $s^*-s$). Indeed, a linear dependence of the gap on this distance ensures that the adiabatic evolution is fast when the gap is large and slows down as it approaches the avoided crossing. Estimating $s^*$ up to an additive accuracy of $O(\delta_s)$ suffices, as the gap scales as $O(g_{\min})$ in the whole of $\mathcal{I}_{s^*}$. 

Thus, $s^*$ needs to be pre-computed up to this accuracy to obtain the correct local schedule before running the adiabatic algorithm. This, in turn, is contingent on computing $A_1$ to an additive accuracy of $O(\delta_s)$. In these regions as well, note the presence of spectral parameters $A_1$, $A_2$ as (multiplicative or additive) factors to $s-s^*$ (or $s^*-s$) in the bounds we obtain on $g(s)$. However, as before, these quantities can be replaced by their known lower/upper bounds, leading to only a $O(\poly(\Delta))$ slowdown. For the Ising Hamiltonian, this is just $O(\poly(n))$.

To summarize, in order to obtain the running time in Theorem \ref{thm:main-result-1} by running AQO, it is crucial to first estimate $A_1$ to accuracy $O(\delta_s)$ in a time at most $T$. This ensures: (i) the avoided crossing is correctly predicted to be within the interval $\mathcal{I}_{s^*}=[s^*-\delta_s, s^*+\delta_s]$, and (ii) this pre-computation requires a time that is less than the overall running time $T$ of the adiabatic algorithm. Consequently, we ask the following question: Given the description of the Hamiltonian $H_z$ (in some form), how hard is it to compute the parameter $A_1$ to the desired additive precision?

Since we are interested in the potential for generic speedups in a purely adiabatic setting, we assume that the adiabatic quantum optimizer does not have access to a digital, gate-based quantum computer. However, it is, in principle, possible to estimate $A_1$ to the desired precision using a classical device. In the next section, we prove that this is computationally hard.

\section{Hardness of predicting the position of the Avoided Crossing}
\label{sec:hardness-avoided-crossing}
In the previous sections, we have demonstrated that adiabatic quantum optimization can find the ground states of any Hamiltonian (having a sufficiently large spectral gap and diagonal in the computational basis) with a quadratic speedup over unstructured ``brute force" classical search. However, prior knowledge of the position of the avoided crossing was necessary for the algorithm to be optimal. Predicting the position of the avoided crossing to any accuracy reduces to the problem of estimating the spectral parameter $A_1$ to the same precision. In this section, we prove that estimating this quantity is computationally hard. Particularly for this section, we change the notation slightly for brevity. Consider any $n$-qubit Hamiltonian $H$ diagonal in the computational basis having distinct eigenlevels with eigenenergies $E_0<E_1<\cdots E_{M-1}$, such that each eigenlevel with eigenvalue $E_k$ is $d_k$-fold degenerate. Then, we relabel the spectral parameter $A_1$ (defined in Eq.~\eqref{eq:spectral-parameters}) as $A_1(H)$ for convenience. That is,
\begin{equation}
A_1(H)=\dfrac{1}{2^n}\sum_{k=1}^{M-1}\dfrac{d_k}{E_k-E_0}.
\end{equation}
We first show that the problem of estimating $A_1$ to within an additive accuracy of $1/\poly(n)$ is \NP-hard, and subsequently prove that estimating it exactly (or to $O(2^{-\poly(n)})$ additive accuracy) is \sharpP-hard.

\subsection{\NP-hardness of estimating $\mathbf{A_1}$ to a low precision}
\label{sec:np-hardness-A1}

Here, we show that estimating the spectral parameter $A_1$ to an additive precision of at most $1/\poly(n)$ is \NP-hard. Note that this is significantly larger than the precision with which $A_1$ needs to be approximated for AQO (which is roughly $2^{-n/2}$). This implies that the problem of predicting the position of the avoided crossing, even to within a precision substantially larger than what is needed for AQO, is already computationally hard. More precisely, suppose there exists a classical procedure that accepts the description of a classical Hamiltonian $H$ (diagonal in the computational basis) and outputs $A_1$ to a precision of $O(1/\poly(n))$. Then, we prove that by making only two calls to this procedure, it is possible to solve 3-\SAT. In such settings, the problem reduces to distinguishing between two promised thresholds of a Hamiltonian $H$, corresponding to the underlying computationally hard problem. 

Suppose we are given a normalized classical Hamiltonian $H$ with the promise that its ground energy satisfies either (i)~$E_0=0$ or (ii)~$\mu_1\leq E_0 \leq 1-\mu_2$, for some $\mu_1, \mu_2$, and the goal would be two disambiguate between (i) and (ii). Note that this is closely related to the \textit{local Hamiltonian problem}, where it is required to determine whether the ground energy is above or below a certain threshold and is known to be \NP-hard for classical Hamiltonians as long as the gap (between the two thresholds) is at least $1/\poly(n)$ \cite{kempe2006complexity}. This particular variant of the local Hamiltonian problem is still \NP-Hard: the 3-\SAT problem can be solved by distinguishing the two thresholds of the ground energy of a Hamiltonian $H$ that corresponds to this problem.    

Now suppose there is a classical procedure $\mathcal{C}_{\varepsilon}$ that accepts the description of a Hamiltonian $H$ and outputs an estimate of $A_1$ to additive accuracy of $1/\poly(n)$. Then we show that if $\mu_1,\mu_2$~scale as $1/\poly(n)$, then by making only two calls to $\mathcal{C}_{\varepsilon}$, it is possible to disambiguate between (i) and (ii), provided $\varepsilon \leq 1/\poly(n)$. In fact, we prove a more general lemma that provides an upper bound on the accuracy $\varepsilon$ as a function of $\mu_1$ and $\mu_2$, such that the disambiguation is possible by making only two calls to $\mathcal{C}_{\varepsilon}$. As mentioned earlier, 3-\SAT\ can be posed in terms of this disambiguation problem. Thus, this implies that 3-\SAT\ can be solved by making only two calls to any procedure estimating $A_1$ to an additive accuracy of most $1/\poly(n)$, which allows us to subsequently conclude that estimating $A_1$ to this precision is \NP-hard. We begin by proving the following general lemma:
~\\
\begin{lemma}
\label{lem:distinguish-ground-energy-general}
Let $\varepsilon,~\mu_1, \mu_2\in (0,1)$. Suppose there exists a classical procedure $\mathcal{C}_{\varepsilon}(\langle H \rangle)$ that accepts the description of a Hamiltonian $H$ and outputs $\widetilde{A}_1(H)$ such that
$$
\left|\widetilde{A}_1(H)-A_1(H)\right|\leq \varepsilon.
$$
Now consider any $n$-qubit Hamiltonian $H$, diagonal in the computational basis such that it has eigenvalues $0\leq E_0< E_1<\ldots < E_{M-1}\leq 1$ such that $M\in \poly(n)$, and eigenvalue $E_k$ has degeneracy $d_k$. Furthermore, suppose that the ground energy $E_0$ of $H$ satisfies the following: Either (i)~$E_0=0$ or (ii)~$0\leq \mu_1\leq E_0\leq 1-\mu_2 \leq 1$. Then, it is possible to decide whether (i) or (ii) holds, by making two calls to $\mathcal{C}_{\varepsilon}$, provided
$$
\varepsilon < \dfrac{\mu_1}{6(1-\mu_1)}-\dfrac{d_0}{6~2^n}\cdot\dfrac{1}{\mu_1\mu_2}
$$
\end{lemma}
\begin{proof}
Given a description of $H$, we aim to disambiguate between $E_0 = 0$ and $\mu_1\leq E_0 \leq 1-\mu_2$ by querying $\mathcal{C}$ twice. First, consider the Hamiltonian 
$$
H'\defeq H\otimes \left(\frac{1+\sigma_z}{2}\right).
$$ 
\textbf{When $\mathbf{E_0=0}$:~} In this case we have,
$$
A_1(H)=\dfrac{1}{2^n}\sum_{k=1}^{M-1}\dfrac{d_k}{E_k}.
$$
Now, the ground energy of $H'$ is zero, with degeneracy $d'_0=d_0+2^n$, while for every other distinct eigenlevel has energy $E'_{k}= E_k$ with degeneracy $d_k$, where $1\leq k\leq M-1$. So,
$$
A_1(H')=\dfrac{1}{2^{n+1}}\sum_{k=1}^{M-1}\dfrac{d_k}{E_k}.
$$
Thus,
$$
A_1(H)-2A_1(H')=0.
$$
\textbf{When $\mathbf{E_0\neq 0}$:~} In this case the ground energy of $H'$ is zero with degeneracy $2^n$ while every other distinct eigenlevel has energy $E'_k=E_{k-1}$ having degeneracy $d'_k=d_{k-1}$, where $1\leq k \leq M$. Therefore,
$$
A_1(H')=\dfrac{1}{2^{n+1}}\sum_{k=0}^{M-1}\dfrac{d_k}{E_k}.
$$
Also,
\begin{align}
A_1(H) &= \dfrac{1}{2^n}\sum_{k=1}^{M-1} \dfrac{d_k}{E_k-E_0} \\
&= \dfrac{1}{2^n}\sum_{k=1}^{M-1}\dfrac{d_k}{E_k} + \dfrac{1}{2^n} \sum_{k=1}^{M-1} \dfrac{d_kE_0}{E_k(E_k-E_0)} \\
&= \dfrac{1}{2^n}\sum_{k=0}^{M-1}\dfrac{d_k}{E_k} - \dfrac{d_0}{2^n E_0} + \dfrac{1}{2^n}\sum_{k=1}^{M-1} \dfrac{d_kE_0}{E_k(E_k-E_0)} \\
&= 2A_1(H') - \dfrac{d_0}{2^n E_0} + \dfrac{1}{2^n}\sum_{k=1}^{M-1} \dfrac{d_k~E_0}{E_k(E_k-E_0)} \\
&\geq 2A_1(H') - \dfrac{d_0}{2^n E_0} + \dfrac{1}{2^n}\sum_{k=1}^{M-1} \frac{d_kE_0}{1-E_0} \\
&= 2A_1(H') - \frac{d_0}{2^n E_0} + \left(1- \dfrac{d_0}{2^n}\right)\frac{E_0}{1-E_0} \\
&= 2A_1(H') + \frac{E_0}{1-E_0} - \dfrac{d_0}{2^n}\left(\dfrac{1-E_0 +E_0^2}{E_0-E_0^2}\right) \\
&\geq 2A_1(H') + \dfrac{\mu_1}{1-\mu_1} - \dfrac{d_0}{2^n}\cdot \frac{1}{\mu_1\mu_2}
\end{align}
Thus in this case, 
$$A_1(H) - 2A_1(H') \geq \frac{\mu_1}{1-\mu_1} - \frac{d_0}{2^n}\cdot \frac{1}{\mu_1\mu_2}.
$$
Now, we shall query $\mathcal{C}$ twice: first with $H$ as input, to obtain $\widetilde{A}_1(H)$ and subsequently, $\widetilde{A}_1(H')$ by providing $H'$ as input. The outputs $\widetilde{A}_1(H)$ and $\widetilde{A}_1(H')$ estimate $A_1(H)$ and $A_1(H')$ respectively, with additive precision $\varepsilon$ (i.e.\ the outputs correspond to $A_1(H)\pm \varepsilon$ and $A_1(H')\pm \varepsilon$). Then in the first case, when $E_0=0$, we have
\begin{equation}
\widetilde{A}_1(H) - 2\widetilde{A_1}(H') \leq 3\varepsilon.
\end{equation}
In the second case, when $E_0\neq 0$:
\begin{equation}
\widetilde{A}_1(H) - 2\widetilde{A}_1(H') \geq \frac{\mu_1}{1-\mu_1} - \frac{d_0}{2^n}\cdot\frac{1}{\mu_1\mu_2} - 3\varepsilon.
\end{equation}
In order to disambiguate between the two cases, we need
\begin{equation}
6\varepsilon < \frac{\mu_1}{1-\mu_1} - \frac{d_0}{2^n}\cdot\dfrac{1}{\mu_1\mu_2},
\end{equation}
which completes the proof.
\end{proof}

We now use Lemma \ref{lem:distinguish-ground-energy-general}  to prove that computing $A_1$ for 3-local Hamiltonians on $n$ qubits up to a precision 
$$
\varepsilon < \frac{1}{72}\cdot\frac{1}{n-1},$$ 
is \NP-hard by reducing 3-\SAT\ to it. Formally, we prove the following theorem:
\begin{theorem}
\label{thm:reduction-sat-ground-energy}
The problem of computing $A_1$ up to a precision 
$$
\varepsilon < \dfrac{1}{72}\cdot\dfrac{1}{n-1},
$$ 
for a $3$-local Hamiltonian on $n$ qubits that satisfies the conditions in Definition \ref{def:prob-Ham} is \NP-hard.
\end{theorem}
\begin{proof}
We consider the 2-local version of 3-\SAT, inspired by the reduction of 3-\SAT\ to \MAXproblem-2-\SAT\ (a variant of the 2-\SAT problem which asks the maximum number satisfying clauses of a given Boolean formula) in Ref.~\cite{garey1976simplified}. Suppose $x_i\in\{0,1\}$ is a binary Boolean variable, and $\bar{x}_i$ be its negation. Consider that we are given some $m$ clauses of the form $a_k\lor b_k \lor c_k$, where each $a_k,b_k,c_k$ is either $x_{l}$ or $\overline{x}_l$ with $0\leq l \leq n-1$. A satisfying assignment makes
\[ \bigwedge_{k=0}^{m-1}a_k\lor b_k \lor c_k \]
true. If $n+m < 15$, use brute-force search. Now assume $n+m \geq 15$. Set 
$$ 
P_{x_l} \defeq \dfrac{I-\sigma_z^{(l)}}{2},~\text{and~} P_{\overline{x}_l} \defeq \dfrac{I+\sigma_z^{(l)}}{2}.
$$  
For each $0\leq k < m$, define the following Hamiltonian:
\begin{align*}
H_k \defeq \; &P_{\overline{a}_k} + P_{\overline{b}_k} + P_{\overline{c}_k} + P_{\overline{x}_{n+k}} \\
    & +P_{a_k}P_{b_k} + P_{a_k}P_{c_k} + P_{b_k}P_{c_k} \\
    & +P_{\overline{a}_k}P_{x_{n+k}} + P_{\overline{b}_k}P_{x_{n+k}} + P_{\overline{c}_k}P_{x_{n+k}}.
\end{align*}

If the $k^\text{th}$ clause is satisfied, then the lowest eigenvalue of $H_k$ is $3$. Otherwise, it is $4$. The largest possible eigenvalue of $H_k$ is $6$. Now consider the Hamiltonian, which acts on $2m+2n$ qubits,
\begin{equation}
H \defeq \frac{1}{6m}\sum_{k=0}^{m-1}H_k + \frac{1}{2n+2m}\sum_{k=n+m}^{2n+2m-1}P_{x_{k}} - \frac{1}{2}I.
\end{equation}
Note that the eigenvalues of $H$ lie between $0$ and $1$. We aim to disambiguate between $E_0 = 0$ and $E_0 \geq \frac{1}{6m}$ using Lemma \ref{lem:distinguish-ground-energy-general}.
Since $d_0\leq 2^{n+m}$ and we can take $\mu_1 = 1/6m$ and $\mu_2 = 1/2$. This requires,
\begin{align}
\dfrac{1}{6}\cdot \dfrac{1}{6m-1} - \dfrac{d_0}{6\cdot 2^{2n+2m}}\cdot 12m &= \dfrac{1}{6}\dfrac{1}{6m-1} - \dfrac{2~d_0~m}{2^{2n+2m}} \\
&\geq \frac{1}{36}\frac{1}{m+n-1} - \frac{2m}{2^{n+m}} \\
&\geq \frac{1}{72}\frac{1}{m+n-1} > \varepsilon.
\end{align}
Here we have used $n+m\geq 15$ to bound 
$$\dfrac{2m}{2^{n+m}} \leq \dfrac{1}{72}\cdot\dfrac{1}{m+n-1}.
$$
\end{proof}

From these reductions, we conclude that estimating the position of the avoided crossing to even a (very low) precision of $1/\poly(n)$ is \NP-hard. In our proof of Lemma \ref{lem:distinguish-ground-energy-general} and subsequently Theorem \ref{thm:reduction-sat-ground-energy}, the modified Hamiltonian $H'$ is 3-local. In the next section, our proof of \sharpP-hardness of (nearly) exactly estimating the position of the avoided crossing works for even 2-local, Ising Hamiltonians.  

\subsection{\sharpP-hardness of estimating $\mathbf{A_1}$ (nearly) exactly}
\label{sec:sharp-P-hardness-A1}
In this section, we prove that the problem of estimating $A_1$ exactly or to a precision of $O(2^{-\poly(n)})$ is \sharpP-hard. We first prove the claim that it is \sharpP-hard to estimate $A_1$ exactly. For this, consider the Ising Hamiltonian 
\begin{equation}
H_{\sigma}=\sum_{\langle i,j \rangle} J_{ij}\sigma^i_z\sigma^j_z+\sum_{j=1}^{n} h_j \sigma^j_z, 
\end{equation}
defined as in Eq.~\eqref{eq:Ising_Ham} (parameters $J_{ij}$ and $h_{i}$ are in some constant range of integers, and the distinct eigenvalues $E_i$ are integers). We show that by making only a polynomial number of calls to an algorithm that exactly computes $A_1(H_{\sigma})$, one can efficiently compute outcome probabilities of IQP circuits, or the degeneracy of the ground state energy of $H_{\sigma}$. Note that both these problems are known to be \sharpP-hard. 

In particular, we note that the problem determining the degeneracy of the ground state can be related to that of counting the number of solutions of \NP-complete problems. For instance, consider the problem of counting the number of satisfying assignments in a 3-\SAT\ formula. While 3-\SAT\ is \NP-complete problem, its counting analogue, \#3-\SAT, which counts the number of satisfying assignments in the underlying Boolean formula, is \sharpP-complete. It is well known that the solution to a 3-\SAT\ problem can be encoded in the ground states of $H_{\sigma}$ \cite{lucas2014ising, choi2011different}. Then, extracting the degeneracy of the ground states of such a Hamiltonian is equivalent to solving \#3-\SAT.

First, let us note that for the Ising Hamiltonian $H_{\sigma}$, if $h_j, J_{jk}$ belong to any set of integers of constant size, there can be almost $M = O(n^2)$ different values for the eigenenergies of $H_\sigma$. Then, our argument boils down to extracting the degeneracies of $d_{k}$ by making $O(\poly(n))$ calls to any algorithm that estimates $A_1(H_\sigma)$. As argued previously,  estimating $d_0$ would imply solving \#3-\SAT, which is \sharpP-complete. Alternatively, the exact knowledge of the values of $d_k$ associated with each gap $\Delta_{k}$ allows us to compute exactly the output probability of an IQP circuit, which is known to be \sharpP-hard. That is, from the knowledge of the degeneracies $\{d_k\}$, we can evaluate
\begin{equation}
 \left |\bra{0}^{\otimes n} \mathcal{C}_{IQP} \ket{0}^{\otimes n}\right |^2= \left|\frac{1}{2^n}\mathrm{Tr}\left[e^{i\frac{\pi}{8}H_{\sigma}}\right]\right|^2  = \left|\frac{1}{2^n}\sum_{k = 0}^{M - 1} d_k e^{i\Delta_k} \right|^2.
\end{equation}

In what follows, we show how to extract the values of $d_k$ from a function that computes $A_1(H)$ exactly for any Ising Hamiltonian, $H_\sigma$. Formally, we prove the following lemma:

\begin{lemma}
\label{lem:exact-degeneracy-hard}
Suppose there exists a classical algorithm $\mathcal{C}(\langle H\rangle)$ which accepts as input the description of an $n$-qubit Hamiltonian $H$ and outputs $A_1(H)$. Furthermore, suppose $H_{\sigma}$ be the $n$-qubit Ising Hamiltonian (with appropriate parameter ranges, as defined in Eq.~\eqref{eq:Ising_Ham}), such that its eigenenergies are $E_0<E_1<\cdots E_{k}$, where $0\leq k \leq M-1$, with known spectral gaps $\Delta_{k}=E_k-E_0$. Then for all $k\in [0, M-1]$, it is possible to estimate the degeneracy $d_k$ of energy eigenvalue $E_k$ by making $O(\mathrm{poly(n)})$ calls to $\mathcal{C}$. 
\end{lemma}
\begin{proof}
Given any $H$, $\mathcal{C}(\langle H \rangle)$ computes $A_1(H)$ exactly. Now consider the $(n+1)$-qubit Hamiltonian  
\begin{equation}
  H'(x) = H \otimes I - I\otimes \frac{x}{2}\sigma_+^{(n + 1)},  
\end{equation} 
where 
$$\sigma_+^{(n + 1)} = \frac{I + \sigma_z^{(n + 1)}}{2},$$
and $x>0$ will be fixed later. Note that
\begin{equation}
A_1\left(H'(x)\right) = \frac{1}{2^{n + 1}} \left(\sum_{k = 1}^{M - 1}\frac{d_k}{\Delta_k} + \sum_{k = 0}^{M - 1}\frac{d_k}{\Delta_k + x/2}\right),
\end{equation}
which can be estimated with a single call to $\mathcal{C}$, using a description of $H'$ as an input (for some fixed $x$). Furthermore, two successive calls to $\mathcal{C}$ using $H$ and $H'(x)$, would allow us to compute
\begin{equation}
\label{eq:func-polynomial}
 f(x) = 2A_1(H'(x)) - A_1(H) = \frac{1}{2^n} \sum_{k = 0}^{M - 1}\frac{d_k}{\Delta_k + x/2}.  
\end{equation}
In what follows, we show that by computing $f(x)$ for $O(\poly(n))$ different values of $x$, it is possible to extract the values of the degeneracies $d_{k}$. To this end, define the polynomial of degree $M-1$ given by:
\begin{equation}
\label{eq:polynomial-interpolation}
    P(x) = \prod_{k = 0}^{M - 1}(\Delta_k + x/2)f(x) = \frac{1}{2^n}\sum_{k = 0}^{M - 1} d_k \prod_{\ell \neq k} (\Delta_\ell + x/2).
\end{equation}
Since the values of the gaps $\Delta_{k}$ are known (we assume them to be integers), we can evaluate $f(x)$ at $M$ values of distinct $x_i>0$ to guarantee that $f(x)$ is finite) and obtain $M$ values $P(x_i)$. For instance, we choose odd integer points $x_i= i$, where $i\in \{ 1, 3, ... , 2M-1\}$, which are sufficient to fully determine $P(x)$. This requires making only $2M=O(\poly(n))$ calls to $\mathcal{C}$.

Using, for example, Lagrange interpolation (which for a polynomial of degree $M-1$ requires $O(\poly(n))$ time) \cite{phillips2003interpolation}, we can fully reconstruct the polynomial $P(x)$, which can now be evaluated at any point. The degeneracies can be extracted since 
\begin{align}
  d_{k} & = \frac{2^n P (-2\Delta_{k})}{\prod_{\ell \neq k} (\Delta_\ell - \Delta_{k})}, ~k\in \{0,..., M-1\}.
  \label{eq:d_alpha}
\end{align} 
\end{proof}
As argued previously, extracting the degeneracies can be used to solve any problem in the \sharpP\ complexity class. Thus, by using $\mathcal{C}$ as a subroutine, it would also be possible to extract the degeneracies in the spectrum of an NP-Complete Hamiltonian (such as 3-\SAT) in polynomial time, thereby efficiently solving a \#P-Hard problem. It is possible to show that this hardness proof is robust to exponentially small additive errors in the computation of $A_1(H)$, which can be shown using Paturi's lemma (Corollary 1 of Ref.~\cite{paturi1992}), which we state below:
 
\begin{lemma}[Corollary 1 of \cite{paturi1992}]
\label{lem:paturi}
Let $P(x)$ be a polynomial of degree, $\mathrm{deg}(P)\leq M$, and $|P(i)|\leq c$ for integers $i\in \{0,1,\cdots,M\}$. Then, 
$$
\forall x\in [0,M],~|P(x)|\leq c\cdot 2^{\mathrm{deg}(P)}.
$$
\end{lemma}
It is easy to see why this lemma is useful for us. By querying the supposed classical algorithm that approximates $A_1(H_\sigma)$, $O(\poly(n))$ times, we can construct a polynomial sampled at integral points in the domain. Then, it is possible to upper bound its error with the polynomial constructed in Eq.~\eqref{eq:polynomial-interpolation} for the entire domain $[0,M]$ by Paturi's lemma. We prove that when this error is sufficiently small ($O(2^{-\poly(n)})$), the degeneracies $d_k$ can still be extracted. We elaborate on this proof next.  
~\\
First assume we have access to a classical algorithm $\mathcal{C}_{\varepsilon}$ which accepts as input the description of some $n$-qubit Hamiltonian, and outputs $\widetilde{A}_1(H)$ such that 
\begin{equation}
\left|\widetilde{A}_1(H)- A_1(H)\right|\leq\varepsilon.     
\end{equation}
Then, as before, by making two calls to $\mathcal{C}_{\varepsilon}$, we can obtain an approximation $\widetilde{f}(x)$ such that $|\widetilde{f}(x)- f(x)|\leq 3 \varepsilon$, where $f(x)$ is as defined in Eq.~\eqref{eq:func-polynomial}. Subsequently, it is possible to define the $M-1$ degree polynomial $\widetilde{P}(x)$  polynomial defined by the $M$ points  $(x_i=i, \widetilde{P}(x_i))$, where $ i\in \{1, 3, ... , 2M-1\}$ are odd integers such that
\begin{equation}
    \widetilde{P}(x_i) = \prod_{k = 0}^{M - 1}(\Delta_k + x_i/2)\widetilde{f}(x_i), 
\end{equation}
which can be constructed by making $2M=O(\poly(n))$ calls to $\mathcal{C}_{\varepsilon}$. Therefore, it follows that
\begin{equation}
   D(x_i )= \left|\widetilde{P}(x_i)-P(x_i)\right|\leq 3 \varepsilon \left(x_i/2 \times \prod_{k=1}^{M-1}  \left(\Delta_{k}+x_i/2\right) \right), 
\end{equation} 
and hence, from Lemma \ref{lem:paturi}, at any point $x\in [1, 2M-1]$ we have that 
\begin{equation}
    |\widetilde{P}(x)-P(x)|\leq  3 \varepsilon \times 2^{M-1}\times \left(x_i/2 \times \prod_{k=1}^{M-1}  \left(\Delta_{k}+x_i/2\right) \right)=O(\varepsilon ~2^{\text{poly}(n)}).   
\end{equation}

It can be seen, using Eq.~\eqref{eq:d_alpha} that for sufficiently small error $\varepsilon= O(2^{-\text{poly}(n)})$ it is possible to ensure that the additive error in the approximation of the degeneracies $d_{k}$ is less than $1/2$, and so their exact values can simply be obtained by rounding up the approximate value to the nearest integer. The derivation above shows that it is \#P-hard to approximate $A_1(H)$ up to additive error $\varepsilon= O(2^{-\text{poly}(n)})$. We summarize this formally via the following lemma:
~\\
\begin{lemma}
\label{lem:approx-degeneracy-hard}
Suppose there exists a classical algorithm $\mathcal{C}_{\varepsilon}(\langle H\rangle)$ which accepts as input the description of an $n$-qubit Hamiltonian $H$, and outputs $\widetilde{A}_1(H)$ such that
$$
\left|\widetilde{A}_1(H)- A_1(H)\right|\leq\varepsilon.     
$$
Furthermore, suppose $H_{\sigma}$ be the $n$-qubit Ising Hamiltonian (with appropriate parameter ranges, as defined in Eq.~\eqref{eq:Ising_Ham}), such that its eigenenergies are $E_0<E_1<\cdots E_{k}$, where $0\leq k \leq M-1$, with known spectral gaps $\Delta_{k}=E_k-E_0$. Then for all $k\in [0, M-1]$, for some sufficiently small $\varepsilon\in O(2^{-\poly(n)})$, it is possible to estimate the degeneracy $d_k$ of energy eigenvalue $E_k$ by making $O(\poly(n))$ calls to $\mathcal{C}_{\varepsilon}$. 
\end{lemma}
Lemma \ref{lem:exact-degeneracy-hard} and Lemma \ref{lem:approx-degeneracy-hard} combined, prove Theorem \ref{thm:main-result-3}.
~\\~\\
Our proof can be extended to the probabilistic case, namely where the classical algorithm $\mathcal{C}_{\varepsilon}$ outputs an estimate $A_1$ to an additive accuracy of $\varepsilon\in O(2^{-\poly(n)})$ with a constant probability (say $\geq 3/4$). Then, even in this setting, the degeneracies $d_k$ of an Ising Hamiltonian $H_{\sigma}$ can be exactly estimated with a high probability. This implies the problem of estimating $A_1$ remains hard even when the classical algorithm is probabilistic. The central ideas are as follows: In this case, two queries to $\mathcal{C}_{\varepsilon}$ are needed to obtain each correct sampling point $(x_i, \widetilde{P}(x_i))$, with a constant probability. As before, we need to re-construct the polynomial $\widetilde{P}$ of degree $M-2$ for distinct samples. However, instead of making $2M$ queries to $\mathcal{C}_{\varepsilon}$, we query the algorithm some $k=O(\poly(n))$ times (such that this is sufficiently larger than $M-2$). In fact, by using the Chernoff bound, it is possible to determine that with a high probability, $(k+M-2)/2$ correct sample points $(x_i,\widetilde{P}(x_i))$ can be obtained by making $k=O(\poly(n))$ queries to the classical algorithm. Then, the Berlekamp-Welch algorithm, which we state below, allows us to reconstruct the polynomial $\widetilde{P}$ exactly:   
~\\
\begin{theorem}[Berlekamp-Welch algorithm \cite{movassagh2023hardness}]
\label{thm:berlekamp-welch}
Let $P$ be a polynomial of degree $d$ over any field $\mathbb{F}$. Suppose we are given $k$ pairs of elements $\{(x_i, y_i)\}_{i=1,\ldots,k}$ such that $y_i=P(x_i)$ for at least $\max\{d+1, (k+d)/2\}$ points. Then, the polynomial $P$ can be recovered in time $O(\poly(k,d))$. 
\end{theorem}
~\\
It is worth mentioning that the Berlekamp-Welch algorithm has been instrumental in proving hardness results in quantum supremacy proposals \cite{movassagh2023hardness, bouland2019complexity}. In our context, this algorithm allows us to reconstruct the polynomial $\widetilde{P}$ exactly (we have $d=M-2$). Subsequently, using the arguments of Lemma \ref{lem:approx-degeneracy-hard} it is possible to prove that the degeneracies $d_k$ can be estimated exactly with a high probability, provided the additive accuracy with which $\mathcal{C}_{\varepsilon}$ estimates $A_1$ is some $\varepsilon\in O(2^{-\poly(n)})$. 

Unfortunately, these proof techniques based on polynomial interpolation do not allow us to conclude anything about the hardness of the approximation of $A_1(H)$ up to the additive error tolerated by the adiabatic algorithm described in Sec.~\ref{sec:poac-and-schedule}.  We remark that limitations of hardness proofs based on polynomial interpolation were also identified in the context of proposals for demonstration of quantum computational advantage such as Boson sampling \cite{aaronson2011bosonsampling} and random circuit sampling \cite{bouland2019complexity}, leading to conjectures about the hardness of approximating outcome probabilities from these circuits which still remain open. From the statement of Theorem \ref{thm:main-result-1}, we require an algorithm that estimates $A_1$ to additive accuracy $O(\delta_s)$ in time at most $\widetilde{O}(\sqrt{2^n/d_0})$. However, the best classical algorithms for approximating the solutions to \sharpP-complete problems leverage sampling-based techniques and have a running time of $O(1/\delta^2_s)=\widetilde{O}(2^n/d_0)$ \cite{gurvitz2005mixed}. One might consider using sampling-based methods to estimate $A_1$ to the desired precision. For instance, given a Hamiltonian $H$, suppose there exists a procedure that efficiently samples the inverse gaps $\Delta_k$ according to the probability distribution $\{d_k/2^n\}_{k=1,\ldots, M-1}$. Then, how many queries to this procedure are required to estimate the expectation value of this distribution (which is $\mu=A_1$) to a precision of $\delta_s$? Note that the variance is $\sigma^2=A_2-A_1^2\leq A_2$. By using Chebyshev's inequality, it is easy to see that at least $\widetilde{O}(2^n/d_0)$ queries are needed to estimate $A_1$ to a precision of $\delta_s$, with a constant probability. While arguing along these lines provides evidence that it is hard to estimate $A_1$ to the desired precision, we could concretely demonstrate that this is \NP-hard from Sec.~\ref{sec:np-hardness-A1}, albeit for a 3-local Hamiltonian.

Overall, these results prove that it is \NP-hard to estimate $A_1$ (consequently, the position of the avoided crossing) approximately, and \sharpP-hard to compute it exactly or nearly exactly. Since the optimality of AQO is contingent on predicting the position of the avoided crossing, our results point to a possible limitation of this approach, which is absent in the circuit model. We leave open the question of whether this can be circumvented without requiring access to a digital quantum computer. 

\section{Discussion}
\label{sec:discussion}
We explore the possibility of obtaining quadratic speedups over brute-force classical search in adiabatic quantum optimization for solving \NP-complete problems. We provide an adiabatic algorithm that can find the minimum of a cost function encoded in a classical Hamiltonian (diagonal in the computational basis). Our method is optimal for any (classical) problem Hamiltonian with a large enough spectral gap. This includes Ising Hamiltonians $H_{\sigma}$, whose ground states are known to encode solutions of \NP-complete problems. In particular, we prove that AQO requires $\widetilde{O}(2^{n/2})$ time to prepare the ground states of $H_{\sigma}$, matching the lower bound of Ref.~\cite{farhi2008fail} up to polylogarithmic factors. In order to derive our results, we obtain tight bounds on the spectral gap throughout the adiabatic evolution and a closed-form expression of a quantity that approximates the position of the avoided crossing. This allowed us to construct the local adiabatic schedule leading to the aforementioned running time.

Some of the ideas presented in this work can also be used to obtain provable runtime guarantees for non-adiabatic, continuous-time quantum algorithms for preparing the ground states of the Ising Hamiltonian, $H_{\sigma}$ \cite{callison2019finding}. For instance, consider the time-independent Hamiltonian
$$
H = -\ket{\psi_0}\bra{\psi_0}+r H_{\sigma},
$$
where $r$ is a parameter that can be independently chosen. It is possible to prove that for a choice of $r$ that is within an additive precision of roughly $\widetilde{O}(2^{-n/2})$ of $A_1$, evolving the equal superposition state $\ket{\psi_0}$, according to $H$, ensures an oscillation between $\ket{\psi_0}$ and the ground state of the problem Hamiltonian. The algorithm fails for any choice of $r$ outside this symmetric interval around $A_1$. Indeed, for the correct choice of $r\approx A_1$, the resulting quantum state has at least a $1/\poly(n)$ overlap with the ground states of $H_{\sigma}$ after a time $\widetilde{O}(2^{n/2})$. The details of these findings will appear elsewhere \cite{eduardo}. However, as with the adiabatic quantum optimization algorithm, it is necessary to compute the spectral parameter $A_1$ to the desired precision before running it. Thus, our results on the computational hardness of estimating $A_1$ exactly or approximately also impact other (non-adiabatic) continuous-time quantum algorithms for solving optimization problems.

A natural question is whether this issue persists for the standard choice of Hamiltonians in adiabatic quantum optimization algorithms, i.e.\ when the adiabatic Hamiltonian interpolates between an initial Hamiltonian that is a local sum of spins, given by \ $H_0=\sum_{j} \sigma^j_x$ (instead of a one-dimensional projector), and an Ising Hamiltonian $H_{\sigma}$. That is, 
\begin{equation}
\label{eq:adiabatic-ham-local-sigma-x}
    H(s)=-(1-s)\sum_j \sigma^j_x + s H_{\sigma}. 
\end{equation} 
When $H_{\sigma}$ corresponds to \NP-complete problems, the spectrum of $H(s)$ is extremely complicated, and analytical results are absent. Some numerical findings have reported that the running time in this setting is substantially better than unstructured search \cite{arthurthesis}. On the other hand, for some random instances of \NP-complete problems, it has been observed that the spectrum of $H(s)$ has exponentially many avoided crossings with exponentially small gaps \cite{altshuler2010anderson}. This makes it significantly more challenging to handle this problem analytically and derive a provable running time. Our results demonstrate that even in a simpler setting (such as unstructured search), where the underlying Hamiltonian is guaranteed to have only a single avoided crossing, the position of this avoided crossing can be hard to approximate.  

Our work points to fundamental open problems concerning AQO. For instance, is it possible to use AQO to obtain generic quadratic speedups for solving optimization problems without needing the assistance of a gate-based, digital quantum computer (or having to solve a computationally hard problem in the process)? One possible direction could be to modify the adiabatic Hamiltonian itself: expanding the dimension of the Hamiltonian (by adding qubits) and introducing intermediate Hamiltonians at various points in the adiabatic path. These strategies can potentially shift the position of the avoided crossing to a point independent of the spectrum of the problem Hamiltonian. Furthermore, our results offer interesting insights into the hardness of predicting the position of the avoided crossing. We demonstrate that this quantity is already \NP-hard\ to estimate to an additive precision that is much lower than what is required by AQO. On the other hand, we also showed that it is \sharpP-hard to exactly (or nearly exactly) estimate this quantity. It would be interesting to explore the precise complexity of estimating the position of the avoided crossing to the desired accuracy.
 
\begin{acknowledgments}
SC acknowledges funding from the Science and Engineering Research Board, Department of Science and Technology (SERB-DST), Government of India under Grant No. SRG/2022/000354, and from the Ministry of Electronics and Information Technology (MeitY), Government of India, under Grant No. 4(3)/2024-ITEA. SC also acknowledges support from Fujitsu Ltd, Japan, and IIIT Hyderabad via the Faculty Seed Grant. AB and LN acknowledge support from the INQA Network for funding AB's visit to LN. AB also thanks INL for its hospitality while visiting LN. LN acknowledges support from the Digital Horizon Europe project FoQaCiA,
GA no. 101070558, and from FCT-Fundação para a Ciência
e a Tecnologia (Portugal) via the Project No.
CEECINST/00062/2018.
AB, JC and JR acknowledge support from the Belgian Fonds de la Recherche Scientifique - FNRS under Grants No. R.8015.21 (QOPT) and O.0013.22 (EoS CHEQS).
LN and AB acknowledge discussions with Adrian Tanasa and Simon Martiel. 
\end{acknowledgments}
\bibliography{bibliography}
\bibliographystyle{unsrturl}

\newpage
\setcounter{equation}{0}
\setcounter{figure}{0}
\setcounter{table}{0}
\setcounter{algocf}{0}
\setcounter{section}{0}
\setcounter{theorem}{0}
\setcounter{definition}{0}
\renewcommand{\theequation}{A\arabic{equation}}
\renewcommand{\thetable}{A\arabic{table}}
\renewcommand{\thefigure}{A\arabic{figure}}
\renewcommand{\thetheorem}{A\arabic{theorem}}
\renewcommand{\thedefinition}{A\arabic{definition}}
\renewcommand{\thelemma}{A\arabic{lemma}}
\renewcommand{\thealgocf}{A\arabic{algocf}}
\renewcommand{\thecorollary}{A\arabic{corollary}}
\renewcommand{\thesection}{A - \Roman{section}}

\appendix
\begin{center}
\textbf{\huge Appendix}
\end{center}
\label{sec:appendix}
In the Appendix, we formally derive some of the unproven claims in the main manuscript and also develop some other results. In Sec.~\ref{sec-app:ge-fee-Hs}, we find closed-form expressions for the ground and first excited energies of $H(s)$ in the vicinity of the avoided crossing. In Lemma \ref{lemma:right-gap-lower-bound}, the expression for the lower bound for the spectral gap $g(s)$ to the right of the avoided crossing depended on the fact that $f(s)$ is a monotonically decreasing function, which we prove in Sec.~\ref{sec-app:f-decreasing}. In Sec.~\ref{sec-app:adiabatic-theorem}, we derive in detail a simplified version of the adiabatic theorem, which requires fewer assumptions on the adiabatic Hamiltonian. Finally, this allows us to obtain the running time of AQO and prove Theorem \ref{thm:main-result-1} in Sec.~\ref{sec-app:main-result-one-proof}.

\section{Ground and first excited energies of the adiabatic Hamiltonian near the avoided crossing}
\label{sec-app:ge-fee-Hs}
In Sec.~\ref{subsec:poac}, we claimed that in the vicinity of the position of the avoided crossing, i.e.\ for any $s\in \mathcal{I}_{s^*}$, the two lowest eigenvalues of $H(s)$ are of the form $\lambda(s)=sE_0+\delta_{\pm} (s)$, such that $\delta_{\pm}(s)$ are the two smallest solutions of Eq.~\eqref{eq:delta-s-full}, given by
\begin{align*}
&\delta_+(s)\in\left((1-\eta)\delta^{+}_{0}(s),~(1+\eta)\delta^{+}_{0}(s)\right),~\text{ and }\\
&\delta_-(s)\in\left((1+\eta)\delta^{-}_{0}(s),~(1-\eta)\delta^{-}_{0}(s)\right),
\end{align*}
where, $\eta = 0.1$, and
$\delta^{\pm}_{0}(s)$ is as defined in Eq.~\eqref{eq:delta_0}. We prove this claim via the following lemma. 

\begin{lemma}
\label{lem:validity-of-approximations}
Let $H_z$ be any $n$-qubit Hamiltonian, diagonal in the computational basis, with distinct eigenvalues $-1\leq E_0< E_1 < E_2<\cdots< E_{M-1}\leq 1$, such that $E_k$ has degeneracy $d_k$, satisfying $\sum_{k=0}^{M-1} d_k=2^n=N$. Also suppose for a small positive constant $c<0.022$, the spectral gap $\Delta$ of $H_z$ satisfies
\begin{equation}\label{eq:hamiltonian-constraint}
\frac{1}{\Delta}\sqrt{\frac{d_0}{A_2N}}<c.
\end{equation}
Furthermore, suppose $A_1,A_2,A_3$ are as defined in Eq.~\eqref{eq:spectral-parameters} for $p=1,2,3$ respectively. Then, for any $s\in \mathcal{I}_{s^*}$, the two smallest solutions of Eq.~\eqref{eq:delta-s-full}, $\delta_{\pm}(s)$ are given by 
\begin{align}
\label{eq:delta_+}
&\delta_+(s)\in\left((1-\eta)\delta^{+}_{0}(s),~(1+\eta)\delta^{+}_{0}(s)\right),~\text{ and }\\
\label{eq:delta_-}
&\delta_-(s)\in\left((1+\eta)\delta^{-}_{0}(s),~(1-\eta)\delta^{-}_{0}(s)\right),
\end{align}
where, $\eta$ is a small constant which can be fixed to 0.1, and
\begin{equation}
\delta^{\pm}_{0}(s)=\dfrac{s}{2A_2}\left[ \left(\dfrac{s}{1-s}-A_1\right) \pm\sqrt{\left(\dfrac{s}{1-s}-A_1\right)^2+4A_2\frac{d_0}{N}}~\right].
\end{equation}
\end{lemma}
\begin{proof}
We first define a function 
\begin{equation}
F(\delta)=-\frac{d_0}{N\delta} -\dfrac{1}{1-s} + \frac{1}{N} \sum_{k=1}^{M-1} \frac{d_k}{s(E_k - E_0) - \delta}
\end{equation}
such that $F(\delta)=0$ is equivalent to Eq.~\eqref{eq:delta-s-full}. Then we consider the transformation $\dfrac{1}{1-x}=1+x+\dfrac{x^2}{1-x}$ for any $x\neq 1$ of the summation term in $F(\delta)$ with $x=\dfrac{\delta}{s(E_k-E_0)}$ to obtain
\begin{equation}
\label{eq:expression-F-delta}
F(\delta)=\frac{1}{\delta}\left[-\frac{d_0}{N}+\frac{\delta}{s}\left(A_1-\frac{s}{1-s}\right)\right.\\\left.+\frac{\delta^2}{s^2}A_2+f(\delta)\right]
\end{equation}
where the remaining terms, 
\begin{equation}
f(\delta)=\frac{\delta^3}{s^3}\frac{1}{N} \sum_{k=1}^{M-1} \frac{d_k}{(E_k - E_0)^3}\frac{1}{1-\frac{\delta}{s(E_k-E_0)}}
\end{equation}

Observe that $\delta_0^\pm$ are the roots of the quadratic equation formed by ignoring the terms in $f(\delta)$, and thus we show that $|f(\delta)|$ is upper-bounded by some small constant so that there exists a positive constant $\eta<0.5$ for which $F$ changes sign within the interval from $(1-\eta)\delta_0^\pm$ to $(1+\eta)\delta_0^\pm$. Demonstrating this would allow us to conclude that the roots of Eq.~\eqref{eq:eigenstate-general-adiabatic-Ham} lie within this interval.

Since $\Delta$ is the spectral gap of $H_z$, we have that for all $k>0$, $E_k-E_0 > \Delta$. So, for $\frac{\delta}{s\Delta} < 1$, we obtain
\begin{equation}
\label{eq:upper-bound-f-delta}
|f(\delta)|\le \frac{\delta^3}{s^3}A_3 \frac{1}{1-\frac{\delta}{s\Delta}}
\end{equation}
Notice that from the expression for $\delta^{\pm}_0(s)$ in the statement of the lemma, and using the fact that for any $s\in\mathcal{I}_{s^*}$, $|s-s^*|\leq \delta_s$, where
$$
\delta_s=\dfrac{2}{(A_1+1)^2}\sqrt{\dfrac{d_0 A_2}{N}}, 
$$ 
we obtain the upper bound
\begin{align}
    |\delta^\pm_0(s)|&\leq\dfrac{s(A_1+1)}{2 A_2 (1-s)}\Bigg[\delta_s+\sqrt{\delta^2_s+(A_1+1)^2(1-s)^2\delta^2_s}\Bigg]\\
    &\leq s \dfrac{(A_1+1)^2}{2A_2}\delta_s \left[\frac{1}{(1-s)(1+A_1)}+\sqrt{\frac{1}{(1-s)^2(1+A_1)^2}+1}~\right]  \\
    &\leq s \sqrt{\frac{d_0}{A_2N}} \left[\frac{1}{1-\frac{\delta_s}{1-s^*}}+\sqrt{\frac{1}{(1-\frac{\delta_s}{1-s^*})^2}+1}~\right] \\
    &\leq s\sqrt{\frac{d_0}{A_2N}} \left(\dfrac{1+\sqrt{1+(1-2c)^2}}{(1-2c)}\right) \label{eq:inter-upper-delta0}\\
    &\leq  s\cdot \Delta \cdot c \cdot\kappa \qquad \qquad\left[~\text{Using Eq.~\eqref{eq:hamiltonian-constraint}}~\right]\label{eq:upper-bound-delta_0}
\end{align}
Here, we used the fact that $\frac{\delta_s}{1-s^*} \leq 2c$ (from the proof of Lemma \ref{lem:spectral-gap-in-robustness-window}) and we introduced
$$
\kappa = \left(\dfrac{1+\sqrt{1+(1-2c)^2}}{(1-2c)}\right),
$$
such that $c\kappa$ with any value of $c<0.2424$ ensures that $c\kappa<1$. This leads to the following condition 
\begin{equation}
\dfrac{\left|\delta^\pm_0(s)\right|}{s\Delta}\le c\kappa.
\end{equation}

\noindent Consequently, substituting this upper bound in Eq.~\eqref{eq:upper-bound-f-delta} allows us to obtain for any $\eta$
\begin{equation}
\left|f(\delta^{\pm}_0(1+\eta))\right|\le\dfrac{\left|\delta^{\pm}_0\right|^3}{s^3}A_3\dfrac{(1+\eta)^3}{1-(1+\eta)c\kappa}.
\end{equation}

\noindent We want to show that for some $\eta \leq 0.5$, $F(\delta_0^\pm(1+\eta)$ is changing sign when $\eta$ changes sign as well. For simplicity, we ignore the $1/\delta$ multiplicative factor outside the square brackets in Eq.~\eqref{eq:expression-F-delta} and use the property that $\delta_0^\pm$ is the root of the quadratic function, to obtain: 
\begin{align*}
    F(\delta_0^\pm(1+\eta)) &= -\frac{d_0}{N} +\frac{\delta_0^\pm(1+\eta)}{s} \left( A_1 - \frac{s}{1-s} \right) +\frac{{\delta_0^\pm}^2(1+\eta)^2}{s^2}A_2+ f(\delta_0^\pm(1+\eta)) \\
    &= \eta\frac{\delta_0^\pm}{s} \left( A_1 - \frac{s}{1-s} \right) +\eta(2+\eta)\frac{{\delta_0^\pm}^2}{s^2}A_2+ f(\delta_0^\pm(1+\eta)) \\
    &= \eta\frac{\delta_0^\pm}{s} \left( -\left(  \frac{s}{1-s}-A_1 \right) + 2\frac{\delta_0^\pm}{s}A_2 + \eta \frac{\delta_0^\pm}{s}A_2 \right)+ f(\delta_0^\pm(1+\eta)) \\
    &= \eta\frac{\delta_0^\pm}{s} \left( \pm\sqrt{\left(\dfrac{s}{1-s}-A_1\right)^2+4A_2\frac{d_0}{N}}+ \eta \frac{\delta_0^\pm}{s}A_2 \right)+ f(\delta_0^\pm(1+\eta)) \\
    &= \pm\eta\frac{\delta_0^\pm}{s}  \sqrt{\left(\dfrac{s}{1-s}-A_1\right)^2+4A_2\frac{d_0}{N}}+ \eta^2 \frac{{\delta_0^\pm}^2}{s^2}A_2 + f(\delta_0^\pm(1+\eta)) 
\end{align*}
Consider $\delta^+_0>0$. Using the bound on $f$, we end up with the following two bounds on $F$:
\begin{align}
    F(\delta_0^+(1+\eta) &\geq \eta\frac{\delta_0^+}{s}  \sqrt{\left(\dfrac{s}{1-s}-A_1\right)^2+4A_2\frac{d_0}{N}}+ \eta^2 \frac{{\delta_0^+}^2}{s^2}A_2 - \dfrac{{\delta_0^+}^3}{s^3}A_3\left(\dfrac{(1+\eta)^3}{1-(1+\eta)\kappa c}\right) \label{eq:boundFpos}\\
    F(\delta_0^+(1-\eta)&\leq -\eta\frac{\delta_0^+}{s}  \sqrt{\left(\dfrac{s}{1-s}-A_1\right)^2+4A_2\frac{d_0}{N}}+ \eta^2 \frac{{\delta_0^+}^2}{s^2}A_2 + \dfrac{{\delta_0^+}^3}{s^3}A_3\left(\dfrac{(1-\eta)^3}{1-(1-\eta)\kappa c}\right) \label{eq:boundFneg}
\end{align}
We observe from Eq.~\eqref{eq:inter-upper-delta0} that $\dfrac{{\delta_0^+}^2}{s^2}\leq \dfrac{d_0}{NA_2}\kappa^2$. So from Eq.~\eqref{eq:boundFpos}, we obtain:
\begin{align*}
    F(\delta_0^+(1+\eta) &\geq \eta\frac{\delta_0^+}{s}  \sqrt{\left(\dfrac{s}{1-s}-A_1\right)^2+4A_2\frac{d_0}{N}}+ \eta^2 \frac{{\delta_0^+}^2}{s^2}A_2 - \dfrac{{\delta_0^+}^3}{s^3}A_3\left(\dfrac{(1+\eta)^3}{1-(1+\eta)\kappa c}\right) \\
    &\geq \frac{\delta_0^+}{s} \left ( 2 \eta \sqrt{A_2\frac{d_0}{N}} - \frac{d_0}{N}\frac{ A_3}{A_2}\dfrac{\kappa^2(1+\eta)^3}{1-(1+\eta)\kappa c}\right) \\
    &\geq \frac{\delta_0^+}{s} \sqrt{A_2\frac{d_0}{N}}\left ( 2 \eta  - \frac{ 1}{\Delta}\sqrt{\frac{d_0}{NA_2}}\dfrac{\kappa^2(1+\eta)^3}{1-(1+\eta)\kappa c}\right) \quad \text{where we used } \Delta A_3 \leq  A_2 \\
    &\geq \frac{\delta_0^+}{s} \sqrt{A_2\frac{d_0}{N}}\left ( 2 \eta  - c\dfrac{\kappa^2(1+\eta)^3}{1-(1+\eta)\kappa c}\right)
\end{align*}
This last inequality is positive whenever
\begin{align}
\label{eq:first-ineq}
    \frac{2\eta}{(1+\eta)^3} \geq \dfrac{c\kappa^2}{1-(1+\eta)\kappa c}
\end{align}
Now on Eq. (\ref{eq:boundFneg}):
\begin{align}
    F(\delta_0^+(1-\eta)&\leq -\eta\frac{\delta_0^+}{s}  \sqrt{\left(\dfrac{s}{1-s}-A_1\right)^2+4A_2\frac{d_0}{N}}+ \eta^2 \frac{{\delta_0^+}^2}{s^2}A_2 + \dfrac{{\delta_0^+}^3}{s^3}A_3\dfrac{(1-\eta)^3}{1-(1-\eta)\kappa c} \\
    &\leq \frac{\delta_0^+}{s} \left ( -2\eta \sqrt{A_2\frac{d_0}{N}}+ \eta^2\sqrt{A_2\frac{d_0}{N}} \kappa + \dfrac{d_0}{NA_2}A_3\dfrac{\kappa^2 (1-\eta)^3}{1-(1-\eta)\kappa c}\right) \\
    &\leq \frac{\delta_0^+}{s}\sqrt{A_2\frac{d_0}{N}} \left ( -\eta (2-\kappa \eta) +\dfrac{c\kappa^2(1-\eta)^3}{1-(1-\eta)\kappa c}\right)
\end{align}
This last inequality is negative whenever
\begin{align}
\label{eq:second-ineq}
    \frac{\eta (2-\kappa \eta)}{(1-\eta)^3}\geq \dfrac{c\kappa^2}{1-(1-\eta)\kappa c}
\end{align}
Similar inequality conditions can also be obtained by considering $\delta^{-}_{0}$. Now observe that if we fix $\eta = 0.1$, the first inequality (Eq.~\eqref{eq:first-ineq}) puts the following constraint on $c$: $c\leq 0.022$.  On the other hand, the inequality in Eq.~\eqref{eq:second-ineq} leads to the upper bound of $c\leq 0.034$. From this, we conclude that choosing any $c\leq 0.022$ suffices in Definition \ref{def:prob-Ham}. Note that other choices of $\eta$, and $c$ are also possible.

Finally, this shows that the roots of $F(\delta)=0$ exist in the intervals 
\begin{align*}
&\delta_+(s)\in\left((1-\eta)\delta^{+}_{0}(s),~(1+\eta)\delta^{+}_{0}(s)\right),~\text{ and }\\
&\delta_-(s)\in\left((1+\eta)\delta^{-}_{0}(s),~(1-\eta)\delta^{-}_{0}(s)\right),
\end{align*}
where
\begin{equation*}
\delta^{\pm}_{0}(s)=\dfrac{s(A_1+1)}{2A_2(1-s)}\left[~s-\dfrac{A_1}{A_1+1}\right.\left.\pm\sqrt{\left(\dfrac{A_1}{A_1+1}-s\right)^2+\dfrac{4A_2d_0}{N(A_1+1)^2}(1-s)^2}~\right]. 
\end{equation*}
This completes the proof.
\end{proof}
\section{Gap to the right of the avoided crossing: Proof that $\mathbf{f(s)}$ is monotonically decreasing}
\label{sec-app:f-decreasing}
In Lemma \ref{lemma:right-gap-lower-bound}, we claimed that the spectral gap of $H(s)$ to the right of the avoided crossing $s^*$ is given by $g(s)\geq 2\beta(s)[1+f(s)]^{-1}$, where the function
\begin{equation}
f(s)=\left[\dfrac{ 1 +\dfrac{4N \beta^2 A_2}{s^2 d_0}}{1+\dfrac{N \beta }{d_0} \left(\dfrac{s-s^*}{s(1-s)(1-s^*)}\right) -\dfrac{2 N \beta^2 A_2}{s^2 d_0}}\right],
\end{equation}
is monotonically decreasing in the interval $[s^*,1]$. Here we show that this is indeed the case for $f(s)$ via the following lemma:
~\\
\begin{lemma}
\label{lem:monotonic-dec-f}
For the choice of $a,~k$ in Lemma \ref{lemma:right-gap-lower-bound}, the function $f(s)$ is monotonically decreasing in the interval $s\in [s^*, 1]$, where $s^*=A_1/(A_1+1)$.
\end{lemma}
\begin{proof}
Recall, that
$$
\beta(s)=a \left(\dfrac{s-s_0}{1-s_0}\right),
$$
where,
$$
s_0=s^*-k~g_{\min}\dfrac{1-s^*}{a-k g_{\min}},
$$ 
and appropriate values of $a,~k$ will be chosen later.

We will show that for $s\in [s^*, 1]$, $f'(s)<0$. For this, we express $f=\frac{u}{v}$ so that $f'=\frac{u'v-uv'}{v^2}$, where $u$ and $v$ are given by:
\begin{align*}
    u&=\dfrac{ s^2 (1-s) d_0}{N}+4A_2\beta^2 (1-s)\\
    v&=\dfrac{ s^2 (1-s) d_0}{N}+\beta s \dfrac{s-s^*}{1-s^*}-2A_2\beta^2  (1-s)
\end{align*}
We have,
\begin{align*}
    u'v=&\left [\dfrac{4aA_2\beta}{1-s_0}(2+s_0-3s) +\dfrac{d_0}{N} s(2-3s)\right ]\times \left [\beta s \left(\dfrac{s-s^*}{1-s^*}\right)-2A_2 \beta^2(1-s) +\dfrac{d_0}{N} s^2 (1-s)\right ] \\
    =&\phantom{+}\frac{4aA_2\beta^2}{1-s_0}(2+s_0-3s)  \left [ s \frac{s-s^*}{1-s^*}-2A_2 \beta(1-s)\right ] \\
    &-  \dfrac{2 A_2 \beta^2 d_0}{N} s(2-3s)(1-s)+\dfrac{d_0\beta}{N} s^2(2-3s) \left(\dfrac{s-s^*}{1-s^*}\right)\\
    &+ \dfrac{4a A_2 \beta~d_0 }{N(1-s_0)} s^2 (1-s) (2+s_0-3s)+  \dfrac{d_0^2}{N^2} s^3(2-3s)(1-s).
\end{align*}
Also,
\begin{align*}
    uv'=&\phantom{+}\left [4A_2 \beta^2 (1-s) +\dfrac{d_0}{N} s^2 (1-s)\right ]\times \left [a\dfrac{(3s^2-2s(s^*+s_0)+s^*s_0)}{(1-s_0)(1-s^*)}-\dfrac{2aA_2\beta}{1-s_0} (2+s_0-3s)+\dfrac{d_0}{N} s(2-3s)\right ] \\
    =&\phantom{+}\dfrac{4aA_2 \beta^2 (1-s)}{(1-s_0)(1-s^*)}(3s^2-2s(s^*+s_0)+s^*s_0) - \dfrac{8aA_2^2\beta^3 }{1-s_0} (1-s) (2+s_0-3s) \\
    &+\dfrac{4 d_0 A_2\beta^2}{N} s(2-3s)(1-s) + \dfrac{d_0}{N} a~s^2 (1-s)\left(\dfrac{3s^2-2s(s^*+s_0)+s^*s_0}{(1-s_0)(1-s^*)}\right) \\
    &- \dfrac{2a~A_2 s^2 \beta d_0}{N} \left(\dfrac{1-s}{1-s_0}\right) (2+s_0-3s) + \dfrac{d_0^2}{N^2} s^3(2-3s)(1-s)
\end{align*}
\noindent Taking $u'v-uv'$, we see that the terms 
$$
\dfrac{d_0^2 }{N^2}s^3(2-3s)(1-s) \text{~and~} \frac{8aA_2^2\beta^3}{1-s_0} (1-s) (2+s_0-3s),
$$  
cancel out. Thus,
\begin{align}
    u'v-uv'=& \phantom{+}\dfrac{4aA_2\beta^2}{(1-s_0)(1-s^*)}\Bigg[s (s-s^*)(2+s_0-3s)  - (1-s)(3s^2-2s(s^*+s_0)+s^*s_0)\Bigg] \nonumber \\
    &+\dfrac{6a A_2d_0\beta}{N(1-s_0)} s(1-s)\Bigg[ s(2+s_0-3s) - (s-s_0)(2-3s)\Bigg] \nonumber \\
    &+ \dfrac{d_0 s^2 a}{N(1-s_0)(1-s^*)}\Bigg[(2-3s)(s-s_0)(s-s^*)-(1-s)(3s^2-2s(s^*+s_0)+s^*s_0)\Bigg] \\
    =& -\dfrac{4aA_2\beta^2}{(1-s_0)(1-s^*)}\Bigg(s^2(1+s_0-s^*)-2s~s_0+s^*s_0\Bigg) \nonumber \\
    &+\dfrac{12 a A_2 d_0 \beta}{N(1-s_0)} s(1-s)^2s_0 \nonumber \\
    &- \dfrac{d_0 s^2 a}{N(1-s_0)(1-s^*)}\Bigg(-s^2(s^*+s_0-1)+2ss_0s^*-s^*s_0\Bigg), 
\label{eq:u'v-uv'}
\end{align}

\noindent So, we have two negative terms and one positive term in Eq.~\eqref{eq:u'v-uv'}. We now prove that the first term is always larger than the second one.
\begin{align*}
    &-\dfrac{4aA_2\beta^2}{(1-s_0)(1-s^*)}\Bigg(s^2(1+s_0-s^*)-2ss_0+s^*s_0\Bigg)+\dfrac{12~a A_2 \beta d_0}{N(1-s_0)}s(1-s)^2 s_0 \\ 
    =& -\dfrac{4aA_2\beta}{1-s_0}\underbrace{\left( \frac{\beta}{1-s^*}\underbrace{(s^2(1+s_0-s^*)-2ss_0+s^*s_0)}_{\text{always positive, and minimum at some }s<s^*}-3s_0 \dfrac{d_0}{N} \underbrace{s(1-s)^2}_{\text{maximum at }s=1/3<s^*}\right)}_{\text{ bounded by the value at }s^*}\\
    \leq & -\dfrac{4aA_2s^*\beta}{1-s_0} \left( a\frac{(s^*-s_0)^2}{1-s_0}-3s_0 \dfrac{d_0}{N} (1-s^*)^2\right).
\end{align*}
Note that we assume that $s^*\geq1/3$, which in turn implies that $A_1 \geq 1/2$, which is always true for any $H_z$ having $d_0< N/2$. Now as 
$$s^*-s_0=k~g_{min}\frac{1-s^*}{a-k~g_{min}},
$$ 
we can choose an $a$ that leaves the term negative. In fact, $f'$ is negative for any 
\begin{align}
    a<\frac{4}{3}k^2\frac{A_1}{A_2}.
\end{align}
Thus, for the choice of $a$ in Lemma \ref{lemma:right-gap-lower-bound}, i.e.\
$$
a=\frac{4}{3}k^2 \Delta,
$$ 
we indeed see that $f'$ for all $s\geq s^*$. Thus, $f$ is monotonically decreasing in this interval and, consequently, is upper bounded by $f(s^*)$.
\end{proof}
\section{Adiabatic Theorem and the running time of AQC}
\label{sec-app:adiabatic-theorem}

In this section, we shall extend the recent results in \cite{cunningham2024eigenpath} to the continuous-time setting, leading to a rigorous version of the adiabatic theorem while requiring minimal assumptions on $H(s)$. These results allow us to obtain the running time for adiabatic quantum computation involving any Hamiltonian that is twice differentiable such that there exists some bound on its spectral gap $g(s)$. Thus, we are able to apply the general results we develop here to determine the running time of the AQO problem. However, our results are more general and much more widely applicable. In this section, we derive the running time of adiabatic quantum computation for a generic $H(s)$ that satisfies the aforementioned conditions. In the next section, we apply the results developed here to obtain the running time of our algorithm (Theorem \ref{thm:main-result-1}). We begin by outlining the assumptions on $H(s)$ for the results to hold. 

\begin{definition} 
\label{def:adiabatic-ham-general}
Suppose that for $s\in [0,1]$, we have a continuous, twice differentiable path of admissible Hamiltonians $H(s)$ such that the relevant eigenstate of $H(0)$ (in most cases the ground state) can be prepared. Then, we assume the existence of the following quantities related to the spectrum of $H(s)$:

\begin{itemize}
\item [(i)] Suppose $\lambda_0:[0,1]\mapsto \mathbb{R}$ is a continuous function such that $\lambda_0(s)$ is an eigenvalue of $H(s)$, for all $s\in [0,1]$.

\item [(ii)] Define a function $g_0:[0,1]\mapsto [0,1]$ such that $g_0(s)\geq g(s)$ for all $s\in [0,1]$.

\item [(iii)] The intersection of $[\lambda_0(s) - g_0(s), \lambda_0(s) + g_0(s)]$ with the spectrum of $H(s)$ is exactly $\{\lambda_0(s)\}$.
\end{itemize}

Furthermore, define $P(s)$ to be the projector on to the eigenspace associated with the eigenvalue $\lambda_0(s)$, and let $Q(s)=I-P(s)$.  
\end{definition}

Having defined the parametrized Hamiltonian, we can now move on to the derivation of the rigorous version of the adiabatic theorem. Note that these results are quite general and require minimal assumptions on $H$. For our results, we simply assume knowledge of $g_0(s)$, which bounds the gap. No detailed knowledge of the gap, the eigenvalue $\lambda_0$, or any other part of the spectrum is required.

\subsection{Adiabatic Eigenpath Traversal}
\label{eq:adiabatic-epath-traversal}
Suppose we have a time-dependent Hamiltonian $H(t)$. Then the evolution of density matrix $\rho$ under $H(t)$ is given by the Liouville–von Neumann equation, given by
\begin{equation}
\od{\rho}{t} = -i[H,\rho(t)].
\end{equation}
Now suppose the time depends on some parameter $s\in [0,1]$, i.e.\ $t = K(s)$. Then, by performing a change of variable, we have a parametrized Hamiltonian $H(s)$ such that the Liouville–von Neumann equation becomes
\begin{equation} \label{adiabaticDifferentialEq}
\od{\rho}{s} = -iK'[H,\rho(s)],
\end{equation}
where $K(s)$ the schedule. Now consider some initial state $\rho(0)$ that has a projection of $P(0)\rho(0)P(0)$ onto the relevant eigenspace of $H(0)$. As we evolve $\rho$ under $H$ according to the schedule $K(s)$, we consider the projection of $\rho(s)$ on to the relevant eigenspace throughout the evolution. Ultimately, we are interested in bounding the fidelity, $\Tr[P(1)\rho(1)]$. In order to understand why this is related to the adiabatic theorem, consider that the initial state $\rho(0)$ is the ground state of $H(0)$ and $P(s)$ is the projection on to the ground space of $H(s)$. Then, $\Tr[P(1)\rho(1)]$ indicates whether the projection of the final state onto the ground space of $H(1)$ is high. This would indeed be the case if the evolution were adiabatic, as the system would continue to remain in the ground state throughout the evolution. As we prove via the following lemma, this leads us towards a rigorous version of the adiabatic theorem. 

\begin{lemma} \label{lemma:errorBound}
Consider any parametrized Hamiltonian $H(s)$ that satisfies Definition \ref{def:adiabatic-ham-general} and let $K: [0,1]\to \mathbb{R}^+$ be a differentiable function such that the derivative $K'$ is absolutely continuous. Then, the evolution of any state $\rho$ under $H$, according to the schedule $K(s)$ satisfies
$$
\varepsilon = 1-\Tr[P(1)\rho(1)],
$$
such that
\begin{equation}
\varepsilon \leq \begin{aligned}[t] &(K'(1))^{-1}\norm{\big[P'(1), (H(1)-\lambda_0(1)\cdot I)^{+}\big]} \\
&+ \int_0^1(K')^{-1}\norm{\big[P', (H-\lambda_0\cdot I)^{+}\big]'}\diff{s} + \int_0^1\big((K')^{-1}\big)'\norm{\big[P', (H-\lambda_0\cdot I)^{+}\big]}\diff{s}.
\end{aligned}
\label{eq:error-fidelity-state-schedule}
\end{equation}
\end{lemma}
\begin{proof}
First, observe that
$$
\varepsilon = 1 - \Tr\big[P(1)\rho(1)\big] = \Tr\big[P(0)\rho(0)\big] - \Tr\big[P(1)\rho(1)\big] = \Big|\Tr\Big[P\rho\Big]\Big|_0^1.
$$ 
So we track the change in fidelity over time, i.e.\ $\Tr\big[P(s)\rho(s)\big]$. To this end, construct a differential equation for $\Tr[P\rho]$ by taking its derivative with respect to $s$, i.e.\ 
$$
\Tr[P\rho]' = \Tr[P'\rho] + \Tr[P\rho'].$$ 
This can be simplified using the fact that $PP'P = 0$ and $QP'Q = 0$~\footnote{We have $P' = (PP)' = P'P + PP'$, so $PP'P = 2PP'P$ and $QP'Q = 0$.}. This yields,
\begin{equation}
\Tr[P\rho'] = -iK' \Tr\big[P[H,\rho(s)]\big] = -iK' \Tr\big[[H,P\rho(s)]\big] = 0.
\end{equation}
As $HP = \lambda_0 P$, we have the following two equations
\begin{align}
\left(H- \lambda_0\cdot I\right)\rho P &= \left(H\rho - \rho H\right)P = \left[H,\rho\right]P \\
P\rho \left(H- \lambda_0\cdot I\right) &= P\left(\rho H - H\rho \right) = -P\left[H,\rho\right].
\end{align}
We then obtain
\begin{align}
\Tr[P\rho]' &= \Tr[P'\rho] + \Tr[P\rho'] \\
&= \Tr[P'\rho] \\
&= \Tr[PP'Q\rho] + \Tr[QP'P\rho] \\
&= \Tr\Big[PP'(H- \lambda_0\cdot I)^{+}(H-\lambda_0\cdot I)\rho\Big] + \Tr\Big[(H-\lambda_0\cdot I)(H-\lambda_0\cdot I)^{+}P'P\rho\Big] \\
&= \Tr\Big[PP'(H-\lambda_0\cdot I)^{+}[H,\rho]\Big] - \Tr\Big[(H-\lambda_0\cdot I)^{+}P'P[H,\rho]\Big]\\
&= \Tr\Big[P'(H-\lambda_0\cdot I)^{+}[H,\rho]\Big] - \Tr\Big[(H-\lambda_0\cdot I)^{+}P'[H,\rho]\Big] \\
&= \Tr\Big[\big[P', (H-\lambda_0\cdot I)^{+}\big][H,\rho]\Big] \\
&= i(K')^{-1}\Tr\Big[\big[P', (H-\lambda_0\cdot I)^{+}\big]\rho'\Big].
\end{align}
Then, solving the differential equation gives
\begin{align}
\Tr[P\rho]\big|_0^1 &= i\int_0^1(K')^{-1}\Tr\Big[\big[P', (H-\lambda_0\cdot I)^{+}\big]\rho'\Big]\diff{s} \\
&= \begin{aligned}[t]&i(K'(1))^{-1}\Tr\Big[\big[P'(1), (H(1)-\lambda_0(1)\cdot I)^{+}\big]\rho(1)\Big] \\
&- i\int_0^1 \left(\left[(K')^{-1}+\big((K')^{-1}\big)'\right]\cdot\Tr\Big[\big[P', (H-\lambda_0\cdot I)^{+}\big]'\rho\Big]\right)\diff{s}\end{aligned}
\end{align}
Finally, the error $\varepsilon$ can be bounded by bounding the absolute value of these three terms separately. 
\begin{align}
\varepsilon &=\Big|\Tr[P\rho]\big|_0^1\Big|\\
		    & \leq \begin{aligned}[t] &(K'(1))^{-1}\norm{\big[P'(1), (H(1)-\lambda_0(1)\cdot I)^{+}\big]} \\
&+ \int_0^1(K')^{-1}\norm{\big[P', (H-\lambda_0\cdot I)^{+}\big]'}\diff{s} + \int_0^1\big((K')^{-1}\big)'\norm{\big[P', (H-\lambda_0\cdot I)^{+}\big]}\diff{s}.
\end{aligned}
\end{align}
\end{proof}

Notice that for the expression in the RHS of the upper bound on $\varepsilon$, there are terms of the form $(H^+)'$ (derivative of the pseudoinverse), for bounded operators $H$, with zero spectral norm. Formally, we prove the following:

\begin{lemma} \label{lemma:pseudoInverseDerivative}
Let $H(s)$ be a path of operators with spectrum $\sigma\big(H(s)\big)$ such that $0\in\sigma\big(H(s)\big)$ and $\min\big|\sigma(H(s))\setminus\{0\}\big| = g(s) \neq 0$ for all $s$. Furthermore, suppose that $P_H$ is the projector on to the ground space of $H$. Then,
$$
(H^+)' = -H^+H'H^+ + P_H'H^+ + H^+P_H'.
$$
\end{lemma}
\begin{proof}
We calculate
\begin{align}
(H^+)' &= \lim_{h\to 0} \frac{H^+(s+h) - H^+(s)}{h} \\
&= \lim_{h\to 0} \frac{Q_H(s)H^+(s+h) - H^+(s)Q_H(s+h)}{h} + \frac{P_H(s)H^+(s+h) - H^+(s)P_H(s+h)}{h}.
\end{align}
We first develop the second part:
\begin{align}
&\lim_{h\to 0} \frac{P_H(s)H^+(s+h) - H^+(s)P_H(s+h)}{h} \\
&= \lim_{h\to 0} \frac{P_H(s)H^+(s+h) - P_H(s+h)H^+(s+h) - H^+(s)P_H(s+h) + H^+(s)P_H(s)}{h} \\
&= \lim_{h\to 0} \frac{P_H(s) - P_H(s+h)}{h}H^+(s+h) + H^+(s)\frac{P_H(s) - P_H(s+h)}{h}.
\end{align}
Taking the limit gives $P_H'H^+ + H^+P_H'$. For the first part, we calculate
\begin{align}
&\lim_{h\to 0} \frac{Q_H(s)H^+(s+h) - H^+(s)Q_H(s+h)}{h}\\ 
&= \lim_{h\to 0} \frac{H^+(s)H(s)H^+(s+h) - H^+(s)H(s+h)H^+(s+h)}{h} \\
&= \lim_{h\to 0} H^+(s)\frac{H(s) - H(s+h)}{h}H^+(s+h) \\
&= \lim_{h\to 0} -H^+(s)\frac{H(s+h) - H(s)}{h}H^+(s+h) \\
&= -H^+H'H^+.
\end{align}
\end{proof}

We can now make use of Lemma \ref{lemma:pseudoInverseDerivative} to obtain upper bounds on the terms (norm of the commutators) in the RHS of Eq.~\eqref{eq:error-fidelity-state-schedule} in terms of the derivatives of $H$ and $g$. Formally, we have

\begin{lemma} \label{lemma:projectorDerivativeBounds}
Suppose $H(s)$ is a parametrized Hamiltonian satisfying the conditions outlined in Definition \ref{def:adiabatic-ham-general}. Then, we have the following
\begin{itemize}
\item[$(1)$]~$\lVert P'\rVert \leq 2 \dfrac{\lVert H'\rVert}{g}$;
\item[$(2)$]~$\lVert P^{\prime\prime}\rVert \leq 8 \dfrac{\lVert H'\rVert^2}{g} + 2 \dfrac{\lVert H^{\prime\prime}\rVert}{g}$;
\item[$(3)$]~$\norm{\big((H-\lambda_0\cdot I)^+\big)'} \leq 6\dfrac{\norm{H'}}{g^2}$;
\item[$(4)$]~$\norm{\big[P', (H-\lambda_0\cdot I)^+\big]} \leq 4\dfrac{\norm{H'}}{g^2}$;
\item[$(5)$]~$\norm{\big[P', (H-\lambda_0\cdot I)^+\big]'} \leq 40\dfrac{\norm{H'}^2}{g^3} + 4\dfrac{\norm{H''}}{g^2}$.
\end{itemize}
\end{lemma}
\begin{proof}

First, we prove $(1)$. Let $\Gamma$ be a circle in the complex plane, centered at the ground energy with radius $g /2$. Then we have the Riesz form of the projector
\[ P = \frac{1}{2\pi i}\oint_\Gamma R_H(z)\diff{z}, \]
where $R_H(z) = \big(z\cdot I - H\big)^{-1}$ is the resolvent of $H$ at $z$. Then $R_H(z)' = R_H(z)H'R_H(z)$ (the derivative is with respect to $s$, not $z$). As $H$ is a normal operator, the norm $\lVert R_H(z)\rVert$ is equal to the inverse of the distance from $z$ to the spectrum $\sigma(H)$. On the circle $\Gamma$ this is equal to $(g/2)^{-1}$ everywhere. We can then approximate
\begin{align*}
\lVert P'\rVert &= \Big\lVert \frac{1}{2\pi i}\oint_\Gamma R_H(z)' \diff{z} \Big\rVert \\
&\leq  \frac{1}{2\pi}\oint_\Gamma \lVert R_H(z)'\lVert \diff{z} \\
&\leq \frac{1}{2\pi}\oint_\Gamma \lVert R_H(z) \rVert\cdot \lVert H' \rVert\cdot \lVert R_H(z)\rVert\diff{z} \\
&= \frac{1}{2\pi} \Big(\frac{2}{g}\Big)^2 \lVert H'\rVert \oint_\Gamma \diff{z} \\
&= \frac{1}{2\pi} \Big(\frac{2}{g}\Big)^2 2\pi \frac{g}{2} \lVert H'\rVert \\
&= 2\frac{\lVert H'\rVert}{g}.
\end{align*}
Similarly, in order to prove $(2)$, we can write
\[ 
P^{\prime\prime} = \frac{1}{2\pi i}\oint_\Gamma 2R_H(z)H'R_H(z)H'R_H(z) + R_H(z)H^{\prime\prime}R_H(z)\diff{z}. 
\]
Estimating this in the same way yields
\[ \lVert P^{\prime\prime}\rVert \leq 8\frac{\lVert H'\rVert^2}{g^2} + 2\frac{\lVert H^{\prime\prime}\rVert}{g}.  \]

Proofs of $(3),~(4)$ and $(5)$ follow from $(1)$ and $(2)$, as well as Lemma \ref{lemma:pseudoInverseDerivative} and the Hellmann-Feynman theorem.

\end{proof}
Now, we are in a position to state the adiabatic theorem and obtain a bound on the running time of adiabatic quantum computing so that the fidelity of the final state with the ground state is at least $1-\varepsilon$. We do so in the next section.
~\\
\subsection{Scaling the derivative of the schedule with the gap}
Recall that adiabatic quantum computation is an interpolation between two Hamiltonians, \(H_0\) and \(H_P\), such that the evolution is governed by the time-dependent Hamiltonian $H(s)=(1-s)H_0+s H_P$. The system begins in the ground state of \(H_0\) and evolves slowly (adiabatically) over a long enough time \(T\) such that, by the end of the evolution, the final state is \(\varepsilon\)-close to the ground state of \(H_P\). The precise meaning of ``a long enough time" $T$ is captured by the rigorous version of the quantum adiabatic theorem. Also, it is well known that the running time $T$ scales as a polynomial of the inverse gap. However, the precise scaling depends on the nature of the schedule. Indeed, the overall running is minimized when the schedule is local, i.e.\ it takes into account the instantaneous gap of $H(s)$, leading to faster evolution in regions where the gap is large while slowing down when the gap shrinks. In other words, the derivative of the schedule scales with the spectral gap $g(s)$ of the adiabatic Hamiltonian.

In this section, we use the machinery developed in Sec.~\ref{eq:adiabatic-epath-traversal} to provide a generic upper bound on the running time $T$ when such a local schedule is used to perform adiabatic quantum computation. More precisely, the following theorem shows that adiabatic evolution under a local schedule for a time $T$ results in a final state that has at least $1-\varepsilon$ with the target state (for instance, the ground state of the problem Hamiltonian).

\begin{theorem} 
\phantomsection \label{theorem:adaptiveRate}
Suppose there exists a Hamiltonian $H(s)$ satisfying Definition \ref{def:adiabatic-ham-general} and a lower bound on its spectral gap $g(s)$, that is absolutely continuous in $s\in [0,1]$. Furthermore suppose that there exists $1\leq p\leq 2$ and $B_1,~B_2$ such that 
$$
\int_0^1 \dfrac{1}{g_0(s)^{p}}\diff{s} \leq B_1~g_{\min}^{1-p},
$$ 
and 
$$
\int_0^1 \dfrac{1}{g_0(s)^{3-p}}\diff{s} \leq B_2~g_{\min}^{p-2},
$$ 
for all instances of the problem. For any
\begin{equation} \label{eq:constant-bound}
c \geq \sup_{s\in[0,1]}\left(4\norm{H'(s)} + 40\norm{H'(s)}^2B_2 + 4\norm{H''(s)} + 6p\cdot|g'(s)|\,\norm{H'(s)}B_2 \right),
\end{equation}
the target state is produced with a fidelity of at least $1-\varepsilon$ if
\begin{equation}
K' = \dfrac{1}{\varepsilon}\cdot\dfrac{c}{g_0(s)^p\cdot  g_{\min}^{2-p}},
\end{equation}
and the time complexity is given by \begin{equation}
\label{eq:adiabatic-time}
T = \int_{0}^{1} K'~\diff{s} \leq \dfrac{1}{\varepsilon}\cdot\dfrac{c\cdot B_1}{g_{\min}}.
\end{equation}
\end{theorem}
\begin{proof}
Let $\varepsilon_0$ be the actual error of the algorithm, so that we need $\varepsilon_0$ to be upper bounded by $\varepsilon$. In this case, Eq.~\eqref{eq:error-fidelity-state-schedule} in Lemma \ref{lemma:errorBound} becomes
\begin{equation}
\varepsilon_0 \leq \begin{aligned}[t] &\varepsilon \cdot c^{-1}g_{\min}^{2-p}g_0(1)^p\norm{\big[P'(1), \left(H(1)-\lambda_0(1)\cdot I\right)^{+}\big]} \\
&+ \varepsilon \cdot c^{-1}g_{\min}^{2-p}\left(\int_0^1~g_0^p \norm{\big[P', \left(H-\lambda_0\cdot I\right)^{+}\big]'}\diff{s} + \int_0^1~\big|\big(g_0^p\big)'\big|\norm{\big[P', \left(H-\lambda_0\cdot I\right)^{+}\big]}\diff{s}\right).
\end{aligned} \label{eq:errorBoundAdaptiveSchedule}
\end{equation}
We bound each of these terms separately using Lemma \ref{lemma:projectorDerivativeBounds}. For the first term, we can make use of $(4)$ from Lemma \ref{lemma:projectorDerivativeBounds} to obtain:
\begin{align}
\varepsilon c^{-1}g_{\min}^{2-p} g_0^{p}\norm{\big[P', \left(H-\lambda_0\cdot I\right)^{+}\big]} &\leq 4\varepsilon c^{-1}g_{\min}^{2-p}\cdot g_0^{p}\frac{\norm{H'}}{g_0^2}\\ 
								 &\leq 4\varepsilon c^{-1}\norm{H'}\\ 
								 &\leq 4\varepsilon c^{-1}\sup_{s\in [0,1]}\norm{H'}.
\end{align}
The term $\norm{\big[P', \left(H-\lambda_0\cdot I\right)^{+}\big]'}$ in Eq.~\eqref{eq:errorBoundAdaptiveSchedule} can be bounded by using $(5)$ from Lemma \ref{lemma:errorBound} as:
$$
\norm{\big[P', \left(H-\lambda_0\cdot I\right)^{+}\big]'} \leq 40\dfrac{\norm{H'}^2}{g_0^3} + 4\dfrac{\norm{H''}}{g_0^2}.
$$
Using this upper bound to evaluate the corresponding integral in Eq.~\eqref{eq:errorBoundAdaptiveSchedule}, thus leads to two terms. We evaluate both of them separately. For the first term, we have
\begin{align}
40~\cdot g_{\min}^{2-p}\int_0^1~g_0^{p}\dfrac{\norm{H^{\prime}}^2}{g_0^3}\diff{s} &\leq 40\sup_{s\in[0,1]}\norm{H'(s)}^2g_{\min}^{2-p}\int_0^1\dfrac{1}{g_0^{3-p}}\diff{s} \\
&\leq 40\sup_{s\in[0,1]}\norm{H'(s)}^2B_2~g_{\min}^{2-p}~g_{\min}^{p-2}\\ 
&= 40\sup_{s\in[0,1]}\norm{H'(s)}^2B_2.
\end{align}
We evaluate the second term as,
\begin{align}
4~g_{\min}^{2-p}\int_0^1~g_0^{p}\dfrac{\norm{H^{\prime\prime}}}{g_0^2}\diff{s} &\leq 4\sup_{s\in[0,1]}\norm{H''(s)}g_{\min}^{2-p}\int_0^1~\dfrac{1}{g_0^{2-p}}\diff{s} \\
&\leq 4\sup_{s\in[0,1]}\norm{H''(s)}g_{\min}^{2-p}~g_{\min}^{p-2}\\ 
&= 4\sup_{s\in[0,1]}\norm{H''(s)}.
\end{align}
So, combining these two integrals, the first term in the second line of Eq.~\eqref{eq:errorBoundAdaptiveSchedule} can be upper bounded as: 
\begin{equation}
\varepsilon \cdot c^{-1}g_{\min}^{2-p}\int_0^1~g_0^p \norm{\big[P', \left(H-\lambda_0\cdot I\right)^{+}\big]'}\diff{s} \leq \varepsilon \cdot c^{-1} \left(40\sup_{s\in[0,1]}\norm{H'(s)}^2B_2+4\sup_{s\in[0,1]}\norm{H''(s)}\right)
\end{equation}
Finally, for the second term in the second line of Eq.~\eqref{eq:errorBoundAdaptiveSchedule}, we once again make use of $(4)$ from Lemma \ref{lemma:projectorDerivativeBounds}. This leads to:
\begin{align}
&g_{\min}^{2-p}\int_0^1\big|\big(g_0^{p}\big)'\big|\norm{\big[P', \left(H-\lambda_0\cdot I\right)^{+}\big]}\diff{s}\\ 
&= g_{\min}^{2-p}\int_0^1 p~g_0^{p-1}\big|g_0'\big|\norm{\big[P', \left(H-\lambda_0\cdot I\right)^{+}\big]}\diff{s} \\
&\leq p~g_{\min}^{2-p}\Big(6\sup_{s\in [0,1]}|g_0'(s)|\,\norm{H'(s)}\Big)\int_0^1\frac{g^{p-1}}{g^2}\diff{s} \\
&= p~g_{\min}^{2-p}\Big(6\sup_{s\in [0,1]}|g_0'(s)|\,\norm{H'(s)}\Big)\int_0^1\frac{1}{g^{3-p}}\diff{s} \\
&\leq p~\Big(6\sup_{s\in [0,1]}|g_0'(s)|\,\norm{H'(s)}\Big)~g_{\min}^{2-p}~B_2~g_{\min}^{p-2} \\
&= p~B_2~\Big(6\sup_{s\in [0,1]}|g_0'(s)|\,\norm{H'(s)}\Big).
\end{align}
Putting everything back into Eq.~\eqref{eq:errorBoundAdaptiveSchedule}, gives us
\begin{align}
\varepsilon_0&\leq \varepsilon\cdot c^{-1}\sup_{s\in[0,1]}\Big(4\norm{H'(s)} + 40\norm{H'(s)}^2B_2 + 4\norm{H''(s)} + 6p~|g_0'(s)|\,\norm{H'(s)}B_2 \Big) \\
&= \varepsilon\cdot c^{-1}\cdot c = \varepsilon.
\end{align}
Finally, the time complexity
\begin{align}
T = K(1) = \int_0^1 K' \diff{s} &= \varepsilon^{-1}\int_0^1 \dfrac{c}{g_0^p~g_{\min}^{2-p}} \diff{s}\\ 
		  &= c\cdot\varepsilon^{-1}g_{\min}^{p-2}\int_0^1 \frac{1}{g_0^{p}} \diff{s}\\ 
		  &\leq \varepsilon^{-1}\cdot c \cdot g_{\min}^{p-2}\cdot B_1\cdot g_{\min}^{1-p}\\ 
		  &=\dfrac{1}{\varepsilon}\cdot \dfrac{c\cdot B_1}{g_{\min}}.
\end{align}
\end{proof}

In the next section, we derive the running time of our AQO algorithm by making use of this lemma.

\section{Proof of Theorem \ref{thm:main-result-1}}
\label{sec-app:main-result-one-proof}
Here, we provide a detailed proof of one of our main results, namely the running time of the unstructured adiabatic quantum optimization algorithm (Theorem \ref{thm:main-result-1}). For this, we borrow the general results obtained in Sec.~\ref{sec-app:adiabatic-theorem}. We begin by restating Theorem \ref{thm:main-result-1}.

\mainresultone*
\begin{proof}

We use Theorem \ref{theorem:adaptiveRate} to obtain $T$. Note that in our case, $\lambda_0(s)$ is the ground state of the adiabatic Hamiltonian $H(s)$. Let us first estimate $B_1$. Observe that for this, we need to obtain an upper bound for the integral $\int_{0}^{1} \diff{s}/g_0(s)^{p}$. We use the lower bounds derived on the gap $g(s)$ in Sec.~\ref{sec:poac-and-schedule}, which, in each of the three regions, can be lower bounded by $g_{\min}$ where
$$
g_{\min}=\dfrac{2A_1}{A_1+1}\sqrt{\dfrac{d_0}{A_2 N}}.
$$
Moreover, in each of the three regions, we obtain linear bounds on the gap $g(s)$, which we shall use to evaluate the aforementioned integral. One important fact from Theorem \ref{theorem:adaptiveRate} is that $g_0(s)$ needs to be absolutely continuous in the entire interval of $s\in [0,1]$. We modify the bounds on $g_0(s)$ so that it is a continuous function in the entire interval without affecting the running time of the AQO algorithm. More precisely, it suffices to shrink the lower bounds on $g(s)$ by a factor of 
$$
b=k\left(\dfrac{2}{1+f(s^*)}\right)=k\left(\dfrac{1-8k^2}{1+4k^2}\right)=\dfrac{1}{10},
$$ (for $k=1/4$ defined in Lemma \ref{lemma:right-gap-lower-bound}) so that $g_0(s)$ is now lower bounded by $k~g_{\min}$ for all $s\in [0,1]$. We state rescaled bounds $g_0(s)$ here. For this, recall: 
\begin{align}
\delta_s&=\dfrac{2}{(A_1+1)^2}\sqrt{\dfrac{d_0 A_2}{N}}\\
		&=2(1-s^*)^2\sqrt{\dfrac{d_0~A_2}{N}}\\
		&=\frac{A_2}{A_1(1+A_1)}g_{\min},
\end{align}
and, 
$$
s_0=s^* - k~g_{\min}\frac{1-s^*}{a-k~g_{\min}},
$$
where $a=4k^2\Delta/3$. Then, the rescaled lower bound for the gap in each of the intervals is:
\begin{equation}
g_0(s) = 
\begin{cases}
b\dfrac{ A_1}{A_2}\left(\dfrac{s^*-s}{1-s^*}\right),& \qquad s\in \mathcal{I}_{s^{\leftarrow}}=\Big[0,~s^*-\delta_s\Big)\\~&~\\
b\cdot g_{\min},& \qquad s\in\mathcal{I}_{s^{*}}=\Big[s^*-\delta_s,~ s^*\Big)\\~&~\\
\dfrac{\Delta}{30}\left(\dfrac{s-s_0}{1-s_0}\right),& \qquad s\in\mathcal{I}_{s^{\rightarrow}}=\Big[s^*,~ 1\Big]
\end{cases}
\end{equation}

Thus, the integration of $g_0(s)^{-p}$ proceeds by parts as follows:

\begin{align}
    \int_0^1 g_0(s)^{-p} ds&\leq\int_0^{s^*-\delta_s} g_0(s)^{-p} ds + \int_{s^*-\delta_s}^{s^*} g_0(s)^{-p} ds + \int_{s^*}^1 g_0^{-p} ds \nonumber \\
    &\leq\left[ \frac{A_2}{b\cdot A_1}(1-s^*)\right]^p\int_0^{s^*-\delta_s} \frac{1}{(s^*-s)^p} ds + \dfrac{1}{b^p}\int_{s^*-\delta_s}^{s^*} g_{\min}^{-p} ds + \left[ \frac{30}{\Delta}(1-s_0)\right]^p\int_{s^*}^1 \frac{1}{(s-s_0)^p} ds \nonumber \\
    &\leq\left[ \frac{A_2}{b\cdot A_1}(1-s^*)\right]^p\int_{\delta_s}^{s^*} \frac{1}{u^p} du + \delta_s~b^{-p}~g_{\min}^{-p} + \left[ \frac{30}{\Delta}(1-s_0)\right]^p\int_{s^*-s_0}^{1-s_0} \frac{1}{u^p} du \nonumber \\
    &\leq\left[ \frac{A_2}{b\cdot A_1(1+A_1)}\right]^p \frac{1}{(p-1)~\delta_s^{p-1}}  + \delta_s~b^{-p}~g_{\min}^{-p} + \left[ \frac{30}{\Delta}(1-s_0)\right]^p \frac{1}{(p-1)(s^*-s_0)^{p-1}} \nonumber \\
    &\leq \frac{1}{b^p}\cdot\frac{1}{p-1}\cdot\frac{A_2}{A_1(1+A_1)} \cdot g_{\min}^{1-p}  + \frac{1}{b^p}\cdot\frac{A_2}{A_1(1+A_1)}\cdot g_{\min}^{1-p} + \frac{1}{p-1} \left( \frac{30}{\Delta}\right)^p \left( \frac{a}{k}\right)^{p-1} (1-s_0)~g_{\min}^{1-p} \nonumber \\
    &\leq \frac{p\times 10^p}{p-1}\cdot\frac{A_2}{A_1(1+A_1)} \cdot g_{\min}^{1-p}  + \frac{1}{p-1} \left( \frac{30}{\Delta}\right)^p \left( \frac{\Delta}{3}\right)^{p-1} \frac{1}{1+A_1}\cdot g_{\min}^{1-p} \nonumber \\
    &\leq g_{\min}^{1-p}\frac{1}{(p-1)(1+A_1)}  \left( \frac{p\times 10^p\times A_2}{A_1}+\frac{3\times 10^p}{\Delta}\right) \\
    &\leq g_{\min}^{1-p}\cdot \left[\frac{(p+3)\times 10^p }{(p-1)(1+A_1)\Delta}\right] 
\end{align}
where we used 
$$s^*-s_0=k~g_{\min}\frac{1-s^*}{a-k~g_{\min}}
$$ 
with $a=\frac{4}{3}k^2\Delta$ and applied with $k=1/4$. Also, for the last inequality, we used $\Delta A_2 \leq A_1$. From Theorem \ref{theorem:adaptiveRate}, we have that 
$$
B_1=O\left(\frac{1}{\Delta(1+A_1)}\right).
$$
Similarly, we can upper bound the integral $\int_{0}^{1} g_0(s)^{p-3}\diff{s}$ to obtain that $B_2$ also scales as
$$
B_2=O\left(\frac{1}{\Delta(1+A_1)}\right).
$$ 
Note that the norms of the derivatives of $H(s)$ are constant, and the derivative of $g$ is also bounded by a constant. Thus, we can take $c=O(B_2)$. Finally, for the running time $T$, we have
$$T=O\left(\dfrac{1}{\varepsilon}\cdot \dfrac{B_1\cdot B_2}{g_{\min}}\right)=O\left(\dfrac{1}{\varepsilon}\cdot\frac{\sqrt{A_2}}{\Delta^2 A_1(1+A_1)}\sqrt{\frac{2^n}{d_0}}\right).$$
\end{proof} 
\end{document}